\documentclass[11pt,letterpaper]{article}
\usepackage[authoryear]{natbib}
\pdfoutput=1
\usepackage[pagebackref]{hyperref}
\def\fastmode{0}
\def\arxivmode{1}
\def\stocmode{0}

\def\showauthornotes{1}
\def\showtableofcontents{0}
\def\showkeys{0}
\def\showdraftbox{0}
\def\showcolorlinks{1}
\def\usemicrotype{1}
\def\showfixme{1}

\def\longform{1}

\usepackage{amssymb}  %

\ifnum\fastmode=0
\usepackage{etex}
\fi

\ifnum\fastmode=0
\usepackage[l2tabu, orthodox]{nag}
\fi

\usepackage{xspace,enumerate}

\usepackage[dvipsnames]{xcolor}

\ifnum\fastmode=0
\usepackage[T1]{fontenc}
\usepackage[full]{textcomp}
\fi

\usepackage[american]{babel}

\usepackage{mathtools}

\usepackage{amsthm}

\newtheorem{theorem}{Theorem}[section]
\newtheorem*{theorem*}{Theorem}

\newtheorem{proposition}[theorem]{Proposition}
\newtheorem*{proposition*}{Proposition}
\newtheorem{lemma}[theorem]{Lemma}
\newtheorem*{lemma*}{Lemma}
\newtheorem{corollary}[theorem]{Corollary}
\newtheorem*{conjecture*}{Conjecture}

\newtheorem*{fact*}{Fact}

\newtheorem*{hypothesis*}{Hypothesis}

\theoremstyle{definition}
\newtheorem{definition}[theorem]{Definition}

\newtheorem{example}[theorem]{Example}

\theoremstyle{remark}

\newtheorem*{claim*}{Claim}

\newtheorem*{remark*}{Remark}

\newtheorem*{observation*}{Observation}

\usepackage[
letterpaper,
top=1in,
bottom=1in,
left=1in,
right=1in]{geometry}

\ifnum\arxivmode=0
\usepackage{newpxtext} %
\usepackage{textcomp} %
\usepackage[varg,bigdelims]{newpxmath}
\usepackage[scr=rsfso]{mathalfa}%
\usepackage{bm} %
\let\mathbb\varmathbb
\fi

\ifnum\arxivmode=1
\usepackage[varg]{pxfonts} %
\usepackage{textcomp} %
\usepackage[scr=rsfso]{mathalfa}%
\usepackage{bm} %
\fi

\ifnum\showkeys=1
\usepackage[color]{showkeys}
\fi

\makeatletter
\@ifpackageloaded{hyperref}
{}
{\usepackage[pagebackref]{hyperref}}

\ifnum\showcolorlinks=1
\hypersetup{
colorlinks=true,
urlcolor=blue,
linkcolor=blue,
citecolor=OliveGreen}
\fi

\ifnum\showcolorlinks=0
\hypersetup{
colorlinks=false,
pdfborder={0 0 0}}
\fi

\newcommand{\doi}[1]{\textsc{doi}:
  \href{http://dx.doi.org/#1}{\nolinkurl{#1}}}

\usepackage{prettyref}

\newcommand{\savehyperref}[2]{\texorpdfstring{\hyperref[#1]{#2}}{#2}}

\newrefformat{eq}{\savehyperref{#1}{\textup{(\ref*{#1})}}}
\newrefformat{lem}{\savehyperref{#1}{Lemma~\ref*{#1}}}
\newrefformat{def}{\savehyperref{#1}{Definition~\ref*{#1}}}
\newrefformat{thm}{\savehyperref{#1}{Theorem~\ref*{#1}}}
\newrefformat{cor}{\savehyperref{#1}{Corollary~\ref*{#1}}}
\newrefformat{cha}{\savehyperref{#1}{Chapter~\ref*{#1}}}
\newrefformat{sec}{\savehyperref{#1}{Section~\ref*{#1}}}
\newrefformat{app}{\savehyperref{#1}{Appendix~\ref*{#1}}}
\newrefformat{tab}{\savehyperref{#1}{Table~\ref*{#1}}}
\newrefformat{fig}{\savehyperref{#1}{Figure~\ref*{#1}}}
\newrefformat{hyp}{\savehyperref{#1}{Hypothesis~\ref*{#1}}}
\newrefformat{alg}{\savehyperref{#1}{Algorithm~\ref*{#1}}}
\newrefformat{rem}{\savehyperref{#1}{Remark~\ref*{#1}}}
\newrefformat{item}{\savehyperref{#1}{Item~\ref*{#1}}}
\newrefformat{step}{\savehyperref{#1}{step~\ref*{#1}}}
\newrefformat{conj}{\savehyperref{#1}{Conjecture~\ref*{#1}}}
\newrefformat{fact}{\savehyperref{#1}{Fact~\ref*{#1}}}
\newrefformat{prop}{\savehyperref{#1}{Proposition~\ref*{#1}}}
\newrefformat{prob}{\savehyperref{#1}{Problem~\ref*{#1}}}
\newrefformat{claim}{\savehyperref{#1}{Claim~\ref*{#1}}}
\newrefformat{relax}{\savehyperref{#1}{Relaxation~\ref*{#1}}}
\newrefformat{red}{\savehyperref{#1}{Reduction~\ref*{#1}}}
\newrefformat{part}{\savehyperref{#1}{Part~\ref*{#1}}}

\newcommand{\Sref}[1]{\hyperref[#1]{\S\ref*{#1}}}

\usepackage{nicefrac}

\ifnum\fastmode=0
\ifnum\usemicrotype=1
\usepackage{microtype}
\fi
\fi

\ifnum\showauthornotes=1
\newcommand{\Authornote}[2]{{\sffamily\small\color{red}{[#1: #2]}}}
\newcommand{\Authornotecolored}[3]{{\sffamily\small\color{#1}{[#2: #3]}}}
\newcommand{\Authorcomment}[2]{{\sffamily\small\color{gray}{[#1: #2]}}}
\newcommand{\Authorstartcomment}[1]{\sffamily\small\color{gray}[#1: }

\newcommand{\Authorfnote}[2]{\footnote{\color{red}{#1: #2}}}
\newcommand{\Authorfixme}[1]{\Authornote{#1}{\textbf{??}}}
\newcommand{\Authormarginmark}[1]{\marginpar{\textcolor{red}{\fbox{\Large #1:!}}}}
\else
\newcommand{\Authornote}[2]{}
\newcommand{\Authornotecolored}[3]{}
\newcommand{\Authorcomment}[2]{}
\newcommand{\Authorstartcomment}[1]{}

\newcommand{\Authorfnote}[2]{}
\newcommand{\Authorfixme}[1]{}
\newcommand{\Authormarginmark}[1]{}
\fi

\ifnum\showfixme=0

\fi

\usepackage{boxedminipage}

\newcommand{\card}[1]{\lvert#1\rvert}

\newcommand{\textparen}[1]{\text{(#1)}}

\ifx\because\undefined
\newcommand{\because}[1]{\textparen{because #1}}
\else
\renewcommand{\because}[1]{\textparen{because #1}}
\fi

\newcommand{\defeq}{\stackrel{\mathrm{def}}=}

\newcommand{\seteq}{\mathrel{\mathop:}=}

\newcommand{\mper}{\,.}

\newcommand\bdot\bullet

\ifx\mathds\undefined %

\else

\fi

\DeclareMathOperator{\Tr}{Tr}

\DeclareMathOperator{\val}{val}

\DeclareMathOperator{\conv}{conv}

\let\varprod\undefined

\DeclareMathOperator*{\varprod}{{\textstyle \prod}}

\newcommand{\R}{\mathbb R}

\newcommand{\problemmacro}[1]{\texorpdfstring{\textup{\textsc{#1}}}{#1}\xspace}

\newcommand{\pnum}[1]{{\footnotesize #1}}

\newcommand{\maxcut}{\problemmacro{max cut}}

\newcommand{\maxthreesat}{\problemmacro{max \pnum{3}-sat}}

\newcommand{\cA}{\mathcal A}

\newcommand{\cF}{\mathcal F}

\newcommand{\cH}{\mathcal H}
\newcommand{\cI}{\mathcal I}

\newcommand{\cP}{\mathcal P}
\newcommand{\cQ}{\mathcal Q}

\newcommand{\cS}{\mathcal S}

\renewcommand{\leq}{\leqslant}
\renewcommand{\le}{\leqslant}
\renewcommand{\geq}{\geqslant}
\renewcommand{\ge}{\geqslant}

\ifnum\showdraftbox=1
\newcommand{\draftbox}{\begin{center}
  \fbox{%
    \begin{minipage}{2in}%
      \begin{center}%
          \Large\textsc{Working Draft}\\%
        Please do not distribute%
      \end{center}%
    \end{minipage}%
  }%
\end{center}
\vspace{0.2cm}}
\else
\newcommand{\draftbox}{}
\fi

\let\epsilon=\varepsilon

\numberwithin{equation}{section}

\newcommand\MYcurrentlabel{xxx}

\newcommand{\MYstore}[2]{%
  \global\expandafter \def \csname MYMEMORY #1 \endcsname{#2}%
}

\newcommand{\MYload}[1]{%
  \csname MYMEMORY #1 \endcsname%
}

\newcommand{\MYnewlabel}[1]{%
  \renewcommand\MYcurrentlabel{#1}%
  \MYoldlabel{#1}%
}

\newcommand{\MYdummylabel}[1]{}

\newcommand{\torestate}[1]{%
  \let\MYoldlabel\label%
  \let\label\MYnewlabel%
  #1%
  \MYstore{\MYcurrentlabel}{#1}%
  \let\label\MYoldlabel%
}

\newcommand{\restatetheorem}[1]{%
  \let\MYoldlabel\label
  \let\label\MYdummylabel
  \begin{theorem*}[Restatement of \prettyref{#1}]
    \MYload{#1}
  \end{theorem*}
  \let\label\MYoldlabel
}

\newcommand{\restatelemma}[1]{%
  \let\MYoldlabel\label
  \let\label\MYdummylabel
  \begin{lemma*}[Restatement of \prettyref{#1}]
    \MYload{#1}
  \end{lemma*}
  \let\label\MYoldlabel
}

\newcommand{\restateprop}[1]{%
  \let\MYoldlabel\label
  \let\label\MYdummylabel
  \begin{proposition*}[Restatement of \prettyref{#1}]
    \MYload{#1}
  \end{proposition*}
  \let\label\MYoldlabel
}

\newcommand{\restatefact}[1]{%
  \let\MYoldlabel\label
  \let\label\MYdummylabel
  \begin{fact*}[Restatement of \prettyref{#1}]
    \MYload{#1}
  \end{fact*}
  \let\label\MYoldlabel
}

\newcommand{\restate}[1]{%
  \let\MYoldlabel\label
  \let\label\MYdummylabel
  \MYload{#1}
  \let\label\MYoldlabel
}

\newcommand{\addreferencesection}{
  \phantomsection
\ifnum\stocmode=0
  \addcontentsline{toc}{section}{References}
\else
  \addcontentsline{toc}{section}{References \hspace*{1in} --------- End of extended abstract ---------}
\fi

}

\newcommand{\e}{\epsilon}

\let\origparagraph\paragraph
\renewcommand{\paragraph}[1]{\vspace*{-7pt}\origparagraph{#1.}}

\allowdisplaybreaks

\usepackage{paralist}

\usepackage{comment}

\usepackage{enumitem}

\newlist{enumerate*}{enumerate*}{1}
\setlist[enumerate*]{label=(\arabic*),
  itemjoin={{, }}, itemjoin*={{, and }}}

\newcommand{\psd}[1]{\mathbb{S}^{#1}_+}

\newcommand{\perfectmatching}[1]{\ensuremath{\textup{PM}}(#1)}

\newcommand*{\setcomp}[2]{\left\{#1\,\middle|\,#2\right\}}

\newcommand*{\size}[1]{\left|#1\right|}

\newcommand{\face}[1]{\left\{#1\right\}}

\newcommand{\feq}{\equiv}
\DeclarePairedDelimiter{\vspan}{\langle}{\rangle}
\DeclarePairedDelimiterXPP{\ispan}[1]{}{\langle}{\rangle}{\sb{I}}{#1}
\newcommand{\same}[1]{\simeq\sb{#1}}

\newcommand*{\Gset}[2]{\left\{#1\,\middle|\,#2\right\}}

\title{The matching problem has no small symmetric SDP}

\author{%
G\'abor Braun
\\
Georgia Tech
\and
Jonah Brown-Cohen
\\
U.C.\ Berkeley
\and
Arefin Huq
\\
Georgia Tech
\and
Sebastian Pokutta\thanks{Research reported in this paper was partially
  supported by NSF CAREER award CMMI-1452463 and NSF grant
  CCF-1255900.}
\\
Georgia Tech
\and
Prasad Raghavendra\thanks{This work is supported by the National Science Foundation under the
NSF CAREER Award CCF-1407779, and the Sloan Fellowship from the Alfred P.\ Sloan Foundation.
}
\\
U.C.\ Berkeley
\and
Aurko Roy\thanks{Research reported in this paper was partially
  supported by NSF grant CMMI-1333789.}
\\
Georgia Tech
\and
Benjamin Weitz\thanks{This work was supported by the National Science Foundation Graduate Research Fellowship Program award DGE 1106400}
\\
U.C.\ Berkeley
\and
Daniel Zink
\\
Georgia Tech
}

\hypersetup{
  pdfauthor = {G\'abor Braun, Jonah Brown-Cohen, Arefin Huq, Sebastian
    Pokutta, Prasad Raghavendra, Aurko Roy, Benjamin Weitz, and Daniel
    Zink},
  pdftitle = {The matching problem has no small symmetric SDP},
  pdfsubject = {{MSC2000 Primary
      90C22; %
      Secondary
      68Q17, %
      05C70  %
  }},
  pdfkeywords = {polyhedral approximation; extended formulation;
    matching problem; symmetric algorithm; semidefinite programming}}

\begin{document}

\maketitle
\draftbox
\thispagestyle{empty}

\begin{abstract}
  \cite{Yannakakis91, Yannakakis88} showed that the matching problem does not have a small
  symmetric linear program. \cite{rothvoss2013matching}
  recently proved that any, not necessarily symmetric, linear program
  also has exponential size. In light of this, it is natural to ask whether the matching
  problem can be expressed compactly in a framework such as
  semidefinite programming (SDP) that is more powerful than linear
  programming but still allows efficient optimization.
We answer this question negatively for symmetric SDPs: any symmetric
SDP for the matching problem has exponential size.

We also show that an $O(k)$-round Lasserre SDP relaxation for the
asymmetric metric traveling salesperson problem yields at least as good an approximation
as any \emph{symmetric}  SDP relaxation of size $n^{k}$.

The key technical ingredient underlying both these results is an upper
bound on the degree needed to derive polynomial identities that hold
over the space of matchings or traveling salesperson tours.

\end{abstract}

\clearpage

\ifnum\showtableofcontents=1
{
\tableofcontents
\thispagestyle{empty}
 }
\fi

\clearpage

\setcounter{page}{1}

\section{Introduction}

In his seminal work, \cite{Yannakakis91, Yannakakis88}
showed that any symmetric linear program for the matching problem has
exponential size.
\cite{rothvoss2013matching} recently showed that one can drop the
symmetry requirement: any linear program for the matching problem has
exponential size. Since it is possible to optimize over matchings
in polynomial time, it follows that there is a gap between
problems that have small linear formulations and problems that allow
efficient optimization.

In light of this gap, it is reasonable to ask whether semidefinite
programming (SDP) can characterize all problems that allow efficient optimization.  Semidefinite programs generalize linear programs
and can be solved efficiently both in theory and practice (see
\cite{VandenbergheBoyd96}). SDPs are the basis of some of the best
algorithms currently known, for example the
approximation of \cite{GoemansWilliamson95} for \maxcut{} .

Following prior work (see for example \cite{GouveiaParriloThomas2011}) we define the size of an SDP formulation as the dimension of the psd cone from which the polytope can be obtained as an affine
slice. Some recent work has shown limits to the power of small SDPs.
\cite{BDP2013, DBLP:journals/mp/BrietDP15} nonconstructively give an exponential lower bound on
the size of SDP formulations for most 0/1 polytopes.
\cite{DBLP:conf/stoc/LeeRS15} give an exponential lower bound for solving the traveling salesperson problem (TSP) and approximating \maxthreesat{}. However the question of
whether the matching problem has a small SDP remains open. We give a
partial negative answer to this question by proving the analog of
Yannakakis's result for semidefinite programs:
\begin{theorem}
Any symmetric SDP for the matching problem has exponential size.
\end{theorem}
As we explain below, the main challenge we faced in obtaining this result was to develop machinery to handle the nontrivial structure of the solution space of matchings. 

Using a similar argument, we also show that for the asymmetric metric traveling salesperson problem the
optimal symmetric semidefinite formulation of a given size is essentially
achieved by the respective level of the Lasserre hierarchy.

\subsection*{Related work}

\ifnum\stocmode=0
Bounding the size of general linear programming formulations for a
given problem was initiated by the seminal paper of
\cite{Yannakakis91, Yannakakis88}.  In Yannakakis's model, a general linear program
for say the perfect matching polytope $\perfectmatching{n}$ consists
of a higher-dimensional polytope $Q \in \R^D$ and a projection $\pi$ such
that $\pi(Q) = \perfectmatching{n}$.  The size of the linear program is
measured as the number of inequalities required to define the polytope $Q$.

\cite{Yannakakis91} characterized the size of
linear programming formulations in terms of the non-negative rank of an
associated matrix known as the \emph{slack matrix}.  Using this
characterization, Yannakakis showed that any \emph{symmetric} linear
program for the matching problem or traveling salesperson problem
requires exponential size.  Roughly speaking, a linear program for the
matching problem is symmetric if for every permutation
$\sigma$ of the vertices in the corresponding graph, there is a permutation
$\tilde{\sigma}$ of the coordinates in $\R^{D}$ that leaves the linear program (and thus the polytope \(Q\)) unchanged.
\fi

\ifnum\stocmode=1
In Yannakakis's model (also referred to as
extended formulations), a general linear program
for say the perfect matching polytope $\perfectmatching{n}$ consists
of a higher-dimensional polytope $Q \in \R^D$ and a projection $\pi$ such
that $\pi(Q) = \perfectmatching{n}$.  The size of the linear program is
measured as the number of constraints defining the polytope $Q$.
A linear program for an optimization problem on graphs, like maximum matching
or the traveling salesperson problem, is said to be \emph{symmetric} if for every permutation
$\sigma$ of the vertices in the graph, there is a corresponding permutation
$\tilde{\sigma}$ of the coordinates in $\R^{D}$ that leaves the LP invariant.
\fi

A natural question that came out of the work of Yannakakis is whether
dropping the symmetry requirement helps much. \cite{KaibelPashkovichTheis10} showed that dropping the symmetry requirement can mean the difference between polynomial and superpolynomial size linear extended formulations for matchings with \(\lfloor\log n\rfloor\) edges in \(K_n\), while \cite{DBLP:journals/mp/Goemans15} and \cite{DBLP:journals/mor/Pashkovich14} showed that it can mean the difference between subquadratic and quadratic size linear extended formulations for the permutahedron.
Nonetheless,
\cite{extform4, DBLP:journals/jacm/FioriniMPTW15} and
\cite{rothvoss2013matching} answered this question negatively for the
TSP and matching problems respectively: any linear extended
formulation of either problem, symmetric or not, has exponential size. In particular, \cite{extform4, DBLP:journals/jacm/FioriniMPTW15} established a \(2^{\Omega(\sqrt{n})}\) lower bound on any linear formulation for the TSP, while 
\cite{rothvoss2013matching} later established a \(2^{\Omega(n)}\) lower bound on any linear formulation for matching, which by a known reduction implies a \(2^{\Omega(n)}\) lower bound for the TSP.
From a
computational standpoint, these are strong lower bounds against
solving the TSP or matching problem exactly via small linear programs.
Subsequently, the framework of Yannakakis has been generalized towards
showing lower bounds even for approximating combinatorial optimization
problems in \cite{BraunFPS12, DBLP:journals/mor/BraunFPS15, chan2013approximate, bravermanmoitra13, DBLP:conf/focs/BazziFPS15}.

For the class of maximum constraint satisfaction problems (MaxCSPs),
\cite{chan2013approximate} established a
connection between lower bounds for general linear programs and lower
bounds against an explicit linear program, namely the hierarchy of \cite{SheraliAdams1990}.  Using that connection, \cite{chan2013approximate} showed that
for every constant $d$ and for every MaxCSP, the $d$-round Sherali-Adams
LP relaxation yields at least as good an approximation as any LP relaxation of
size $n^{d/2}$.  By appealing to lower bounds on Sherali-Adams
relaxations of MaxCSPs in literature, they give super-polynomial
lower bounds for \maxthreesat and other MaxCSPs.

Given the general LP lower bounds, it is natural to ask whether the situation is different for SDP relaxations.
Building on the approach of \cite{chan2013approximate},
\cite{Lee2014power}
showed that for the class of MaxCSPs, for every constant \(d\), the $d$-round Lasserre SDP relaxation yields at least as
good an approximation as any symmetric SDP relaxation of size $n^{d/2}$.  In
light of known lower bounds for Lasserre SDP relaxations of
\maxthreesat, this yields a corresponding lower bound for
approximating \maxthreesat.  In a recent advance,
\cite{DBLP:conf/stoc/LeeRS15} show an exponential lower bound even for
\emph{asymmetric} SDP relaxations of the TSP.

\subsection*{Contribution}

We show that there is no small symmetric SDP for the matching problem. 
Our result is an SDP analog
of the result in \cite{Yannakakis91, Yannakakis88} that ruled out a small symmetric LP for the  matching problem.  Specifically, we show:

\begin{theorem}
There exists an absolute constant $\alpha > 0$ such that for every $\e \in
[0,1)$, every \emph{symmetric} SDP relaxation approximating the perfect
matching problem within a factor $1-\frac{\e}{n-1}$ has size at least
$2^{\alpha n}$.
\end{theorem}

To prove this we show that if the matching problem has a small symmetric SDP relaxation, then there is a low-degree sum of squares refutation of the existence of a perfect matching in an odd clique, which contradicts a
result by \cite{grigoriev2001linear}.

The key technical obstacle in adapting the MaxCSP argument to the matching
problem is the non-trivial algebraic structure of the underlying
solution space (the space of all perfect matchings).
A multilinear polynomial is zero over the
solution space of a MaxCSP (the Boolean hypercube $\{0,1\}^n$) if and only if all the coefficients of the polynomial are zero, which is trivial to test.
In contrast, simply testing whether a multilinear polynomial is zero over all perfect matchings is non-trivial.
Nonetheless, in a key lemma we show that every multilinear polynomial $F$ that is identically zero over perfect matchings can be certified as such via a derivation of degree
$2\cdot\deg(F)-1$, starting from the linear and quadratic
constraints that define the space of perfect matchings.

Our second result concerns the asymmetric metric traveling salesperson problem, which is a restriction of the TSP to the case where edge costs obey the triangle equality but are not necessarily symmetric, so that the cost from \(u\) to \(v\) may not equal the cost from \(v\) to \(u\). We show that Lasserre SDP relaxations are more or less optimal
among all symmetric SDP relaxations for approximating the
asymmetric metric traveling salesperson problem. The precise statement follows.
\begin{theorem}
  For every constant $\rho > 0$, if there exists a symmetric SDP relaxation of size
  $r < \sqrt{\binom{2n}{k}} - 1$
  which achieves a $\rho$-approximation for asymmetric metric TSP instances on $2n$
  vertices, then the $(2k-1)$-round Lasserre relaxation achieves a
  $\rho$-approximation for asymmetric metric TSP instances on $n$ vertices.
\end{theorem}

\section{Symmetric SDP formulations}

\label{sec:sym-sdp}
In this section we define a framework for symmetric semidefinite
programming formulations and show that a symmetric SDP formulation implies a symmetric sum of squares representation over a small basis.
Our framework extends the one in
\cite{DBLP:conf/stoc/BraunPZ15} with a symmetry condition; see also
\cite{Lee2014power}.

We first introduce some notation we will use. Let \([n]\) denote the set \(\{1, \ldots, n\}\).
Let \(\psd{r}\) denote the cone of \(r \times r\) real symmetric
positive semidefinite (psd) matrices. Let
\(\R[x]\) denote the set of polynomials in \(n\) real variables \(x = (x_1,\ldots,x_n)\) with real coefficients. For a set \(\cH \subseteq \R[x]\) let \(\vspan{\cH}\) denote the vector space spanned by \(\cH\)
and let \(\ispan{\cH}\) denote the ideal generated by \(\cH\). (Recall that a polynomial \emph{ideal} in a polynomial ring \(R\) is a set that is closed under addition of polynomials in the ideal and closed under multiplication by polynomials in the ring.)

Suppose group \(G\) acts on a set \(X\). The (left) action of \(g \in G\) on \(x \in X\) is denoted \(g \cdot x\). Recall that the \emph{orbit} of \(x \in X\) is \(\setcomp{g \cdot x}{g \in G}\) while the \emph{stabilizer} of \(x\) is \(\setcomp{g \in G}{g \cdot x = x}\). Let \(A_n\) denote the alternating group on \(n\) letters (the set of even permutations of \([n]\)).

We now present our SDP formulation framework. We restrict ourselves to maximization problems
even though the framework extends to minimization problems.
A \emph{maximization problem} \(\cP = (\cS, \cF)\)
consists of a set
\(\cS\) of feasible solutions and a set \(\cF\) of
objective functions.
Suppose we are given two functions \(\tilde{C} \colon \cF \rightarrow \R\) and \(\tilde{S} \colon \cF \rightarrow \R\).
We say an algorithm \((\tilde{C}, \tilde{S})\)-approximately solves \(\cP\)
if for all \(f \in \cF\) with \(\max_{s \in \cS} f(s) \leq \tilde{S}(f)\)
it computes \(\tilde{f}\in \R\) satisfying
\(\max_{s \in \cS} f(s) \leq \tilde{f} \leq \tilde{C}(f)\).
We will refer to \(\tilde{C}\) and \(\tilde{S}\) as the \emph{approximation
guarantees}.

\begin{example}
Suppose we are given a polytope \(P\) with non-redundant inner and outer descriptions:
\[
P = \conv(V) = \setcomp{x}{a_jx \le b_j, j \in [m]}\,.
\]
Let us define \(f_j(x):= b_j-a_jx\) for each \(j \in [m]\).
We can now associate a maximization problem with this polytope by setting
\(\cS = V\) and \(\cF = \setcomp{f_j}{j \in [m]}\), so that each vertex is a solution and the slack with respect to each facet is a function. In order to recover the polytope exactly we would set
\[
\tilde{C}(f) = \tilde{S}(f) = \max_{x \in P} f(x) = 0
\]
for all \(f \in \cF\).
\end{example}

Let \(G\) be a group with associated actions on \(\cS\) and \(\cF\).
The problem \(\cP\) is \emph{\(G\)-symmetric} if the group action
satisfies the compatibility constraint
\((g \cdot f)(g \cdot s) = f(s)\).
For a \(G\)-symmetric problem we require \(G\)-symmetric
approximation guarantees: \(\tilde{C}(g\cdot f) =
\tilde{C}(f)\) and \(\tilde{S}(g\cdot f) = \tilde{S}(f)\) for all \(f \in \cF\) and \(g \in G\). 

We now define the notion of a semidefinite programming
formulation of a maximization problem.
\begin{definition}[SDP formulation for \(\cP\)] Let \(\cP = (\cS,
  \cF)\) be a maximization problem with approximation guarantees \(\tilde{C},
  \tilde{S}\).
  A
  \emph{\((\tilde{C}, \tilde{S})\)-approximate SDP formulation of \(\cP\)}
  of size \(d\)
  consists of a linear map \(\cA \colon \psd{d} \rightarrow \R^k\)
  and \(b \in \R^k\)
  together with
  \begin{enumerate}
  \item \emph{Feasible solutions:}
    an \(X^s \in \psd{d}\) with \(\cA(X^s) = b\)
    for all \(s \in \cS\), i.e., the SDP
    \(\Gset{X \in \psd{d}}{\cA(X) = b}\)
    is a relaxation of \(\conv\face{X^s \mid s \in \cS}\),
  \item \emph{Objective functions:}
    an affine function \(w^f \colon \psd{d} \rightarrow \R\)
	satisfying
    \(w^f(X^s) = f(s)\) 
    for all \(f \in \cF\) with \(\max_{s \in \cS} f(s) \leq \tilde{S}(f)\) and 	all \(s \in \cS\),
    i.e., the linearizations are exact on solutions, and
  \item \emph{Achieving guarantee:}
    \(\max \face{ w^f(X) \mid \cA(X) = b, X \in \psd{d}} \leq \tilde{C}(f)\)
    for all \(f \in \cF\) with \(\max_{s \in \cS} f(s) \leq \tilde{S}(f)\).
  \end{enumerate}
  If \(G\) is a group, \(\cP\) is \(G\)-symmetric,
  and \(G\) acts on \(\psd{d}\),
  then an SDP formulation of \(\cP\)
  with symmetric approximation guarantees \(\tilde{C}, \tilde{S}\)
  is \emph{\(G\)-symmetric}
  if it additionally satisfies
the compatibility conditions for all \(g \in G\): 
\begin{enumerate}
\item \emph{Action on solutions:} \(X^{g \cdot s} = g \cdot X^s\) for
  all \(s \in \cS\).
\item \emph{Action on functions:} \(w^{g \cdot f}(g \cdot X) =
  w^f (X)\) for all \(f \in \cF\) with \(\max_{s
      \in \cS} f(s) \leq \tilde{S}(f)\). 
\item \emph{Invariant affine space:} \(\cA(g  \cdot X) =
  \cA(X)\). 
\end{enumerate}
A \(G\)-symmetric SDP formulation is
\emph{\(G\)-coordinate-symmetric}
if the action of \(G\) on \(\psd{d}\)
is by permutation of coordinates:
that is, there is an action of \(G\) on \([d]\)
with \((g \cdot X)_{ij} = X_{g^{-1} \cdot i, g^{-1} \cdot j}\) for all \(X \in \psd{d}\),
\(i, j \in [d]\) and \(g \in G\).
\end{definition}

We now turn a \(G\)-coordinate-symmetric SDP formulation
into a symmetric sum of squares representation over a small set of basis functions.

\begin{lemma}[Sum of squares for a symmetric SDP
  formulation] \label{lem:factorSymSDP} 
  If 
  a \(G\)-symmetric maximization problem
  \(\cP = (\cS, \cF)\)
  admits a \(G\)-coordinate-symmetric \((\tilde{C}, \tilde{S})\)-approximate
  SDP formulation of size \(d\),
  then there is a set \(\cH\) of
  at most \(\binom{d+1}{2}\) functions
  \(h \colon \cS \to \R\)
  such that for any \(f \in \cF\) with \(\max f \leq \tilde{S}(f)\)
  we have \(\tilde{C}(f) - f = \sum_{j} h_{j}^{2} + \mu_{f}\)
  for some \(h_{j} \in \vspan{\cH}\) and constant \(\mu_{f} \geq 0\).
  Furthermore the set \(\cH\) is invariant under the action of \(G\) given
  by \((g \cdot h) (s) = h (g^{-1} \cdot s)\)
  for \(g \in G\), \(h \in H\) and \(s \in S\).
\end{lemma}

\ifnum\longform=1
	\ifnum\longform=0
	\section{Sum of squares for symmetric SDP formulations}
	\label{sec:sos-sym-sdp-proof}
	In this section we prove \prettyref{lem:factorSymSDP}.
\fi

\begin{proof}
For any psd matrix \(M\) let \(\sqrt{M}\) denote the unique psd matrix
with \(\sqrt{M}^{2} = M\). Note that \(\sqrt{M} \sqrt{M}^{\intercal} = M\) also, since \(\sqrt{M}\) is symmetric.

Let \(\cA\), \(b\), \(\setcomp{X^{s}}{s \in \cS}\), \(\setcomp{w^{f}}{f \in \cF}\) comprise
a \(G\)-coordinate-symmetric SDP formulation of size \(d\).
We define the set \(\cH \coloneqq \Gset{h_{ij}}{i,j \in [d]}\)
via \(h_{ij}(s) \coloneqq \sqrt{X^{s}}_{ij}\).
By the action of \(G\) and the uniqueness of the square root,
we have \(g \cdot h_{ij} = h_{g \cdot i, g \cdot j}\), so \(\cH\) is \(G\)-symmetric.
As \(h_{ij} = h_{ji}\), the set \(\cH\) has at most
\(\binom{d+1}{2}\) elements.

By standard strong duality arguments as in \cite{DBLP:conf/stoc/BraunPZ15},
for every \(f \in \cF\) with \(\max f \leq \tilde{S}(f)\),
there is a \(U^{f} \in \psd{d}\) and \(\mu_{f} \geq 0\)
such that for all \(s \in \cS\),
\[\tilde{C}(f) - f(s) = \Tr[U^{f} X^{s}] + \mu_{f}.\]
Again by standard arguments the trace can be rewritten as a sum of
squares:
\begin{equation*}
\Tr[U^{f} X^{s}]
= \Tr\left[\left(\sqrt{U^{f}}\sqrt{X^{s}}\right)^\intercal
\left(\sqrt{U^{f}}\sqrt{X^{s}}\right)\right]
= \sum_{i,j \in [d]} \left(
    \sum_{k \in [d]} \sqrt{U^{f}}_{ik} \cdot \sqrt{X^{s}}_{kj} \right)^{2}.
\end{equation*}
Therefore
\(\tilde{C}(f) - f = \sum_{i,j \in [d]} \left(
  \sum_{k \in [d]} \sqrt{U^{f}}_{ik} \cdot h_{kj}
\right)^{2}
+ \mu_{f}\),
as claimed.
\end{proof}

\else
	The proof appears in \prettyref{sec:sos-sym-sdp-proof}.
\fi

\section{The perfect matching problem}

We now present the \emph{perfect matching problem}
\(\perfectmatching{n}\)
as a maximization problem in the framework of
\prettyref{sec:sym-sdp} and show that any symmetric SDP formulation
has exponential size.

Let \(n\) be an even positive integer, and
let \(K_{n}\) denote the complete graph on \(n\) vertices.
The feasible solutions of \(\perfectmatching{n}\) are all the
perfect matchings \(M\) on \(K_{n}\).
The objective functions \(f_F\) are indexed by the edge sets \(F\) of \(K_{n}\): for each \(F \subseteq \binom{n}{2}\) we define \(f_F(M) \coloneqq \size{M \cap F}\).
For approximation guarantees we use \(\tilde{S}(f) \coloneqq \max f\)
and \(\tilde{C}(f) \coloneqq \max f + \varepsilon / 2\)
for some fixed \(0 \leq \varepsilon < 1\)
as in \cite{DBLP:conf/soda/BraunP15}; see also \cite{DBLP:journals/tit/BraunP15} for a more in-depth discussion.

\subsection{Symmetric functions on matchings are juntas}
\label{sec:theta-main}
In this section we show that functions on perfect matchings with high symmetry
are actually \emph{juntas}:
they depend only on the edges of a small vertex set.
The key is the following lemma
stating that perfect matchings coinciding on a vertex set
belong to the same orbit of the pointwise stabilizer of the vertex set.
Let \(A_n\) denote the alternating group on \(n\) letters, and for any subset \(X \subseteq [n]\) let \(A(X)\) denote the alternating group that operates on the elements of \(X\) and fixes the remaining elements of \([n]\).
For any set \(W \subseteq [n]\) let \(E[W]\)
denote the edges of \(K_{n}\)
with both endpoints in \(W\).

\begin{lemma}
\label{lem:evenPerm}
  Let $S\subseteq [n]$ with $\size{S} < n/2$ and
  let $M_1$ and $M_2$ be perfect matchings in $K_n$.
  If $M_{1} \cap E[S] = M_{2} \cap E[S]$ then
  there exists $\sigma\in A([n] \setminus S)$ such that
  $\sigma \cdot M_{1} = M_{2}$.
\end{lemma}
\begin{proof}
Let \(\delta(S)\) denote the edges with exactly one endpoint in \(S\).
There are three kinds of edges: those in \(E[S]\), those in \(\delta(S)\), and those disjoint from \(S\). We construct \(\sigma\) to handle each type of edge, then fix \(\sigma\) to be even.

To handle the edges in \(E[S]\) we set \(\sigma\) to the identity on \(S\), since \(M_1 \cap E[S] = M_2 \cap E[S]\).

To handle the edges in \(\delta(S)\) we note that for each edge \((s,v) \in M_1\)
with \(s \in S\) and \(v \notin S\)
there is a unique edge \((s,w) \in M_2\) with
\(w \notin S\). We extend \(\sigma\) to map \(v\) to \(w\) for each such \(s\).

To handle the edges disjoint from \(S\), we again use the fact that \(M_1\) and \(M_2\) are perfect matchings, so the number of edges in each that are disjoint from \(S\) is the same. We extend \(\sigma\) to be an arbitrary bijection on those edges.

We now show that we can choose \(\sigma\) to be even.
Since \(\size{S} < n/2\) there is an edge \((u,v) \in M_2\) disjoint from \(S\). Let $\tau_{u,v}$ denote the
transposition of $u$ and $v$ and
let $\sigma' \coloneqq  \tau_{u,v} \circ \sigma$.
We have $\sigma' \cdot M_1 = \sigma \cdot M_1 = M_2$,
and either $\sigma$ or $\sigma'$ is even.
\end{proof}

We also need the following lemma, which has been
used extensively for symmetric linear extended formulations. See references \cite{Yannakakis88,Yannakakis91,KaibelPashkovichTheis10,BP2011,Lee2014power} for examples.

\begin{lemma}[{\cite[Theorems~5.2A and 5.2B]{DixonBook}}]
\label{lem:DixonMortimer}
Let \(n \ge 10\) and let \(G \leq A_n\) be a group. If \(\card{A_n : G} < \binom{n}{k}\)
for some \(k < n/2\), then
there is a subset \(W \subseteq [n]\) such that \(\size{W} < k\),
  \(W\) is \(G\)-invariant, and
  \(A([n] \setminus W)\) is a subgroup of \(G\).
\end{lemma}

We now formally state and prove the claim about juntas:

\begin{proposition}
  \label{prop:matching-junta}
  Let \(n \geq 10\), let \(k < n/2\)
  and let \(\cH\) be an \(A_{n}\)-symmetric set of functions
  on the set of perfect matchings of \(K_{n}\)
  of size less than \(\binom{n}{k}\).
  Then for every \(h \in \cH\) there is a vertex set
  \(W \subseteq [n]\) of size less than \(k\)
  such that \(h\) depends only on the
  (at most \(\binom{k-1}{2}\))
  edges in \(W\).
\begin{proof}
Let \(h \in \cH\), let \(\operatorname{Stab}(h)\) denote the stabilizer of \(h\), and let \(\operatorname{Orb}(h)\) denote the orbit of \(h\). Since \(\cH\) is \(A_n\)-symmetric we have \(\card{\operatorname{Orb}(h)} < \binom{n}{k}\). By the orbit-stabilizer theorem it follows that
\(\card{A_n : \operatorname{Stab}(h)} < \binom{n}{k}\).
Applying \prettyref{lem:DixonMortimer} to
the stabilizer of \(h\),
we obtain a subset \(W \subseteq [n]\) of size
less than \(k\)
such that \(h\) is stabilized by \(A([n] \setminus W)\),
i.e.,
\[h(M) = (g \cdot h) (M) = h (g^{-1} \cdot M)\]
for all \(g \in A([n] \setminus W)\). 

Therefore for every perfect matching \(M\)
the function \(h\) is constant on the
\(A([n] \setminus W)\)-orbit of \(M\).
As the orbit is determined by \(M \cap E[W]\)
by \prettyref{lem:evenPerm},
so is the function value \(h(M)\).
Therefore \(h\) depends only
on the edges in \(E[W]\).
\end{proof}
\end{proposition}

\subsection{The matching polynomials}

A key step in proving our lower bound is obtaining low-degree derivations of
approximation guarantees for objective functions of
\(\perfectmatching{n}\).
Therefore we start with a standard representation of functions as
polynomials. We define the \emph{matching constraint polynomials} \(\cP_{n}\) as:
\begin{equation}
  \label{eq:matchingconstraint}
  \begin{split}
  \mathcal{P}_{n} \coloneqq &\face{x_{uv} x_{uw} \mid u,v,w \in
  [n] \text{ distinct}} \\
  &\cup \face{\sum_{u \in [n], u \neq v} x_{uv} -
  1 \mathrel{}\middle|\mathrel{} v \in [n]} \\
  &\cup \face{x_{uv}^{2} - x_{uv} \mid u,v \in
  [n] \text{ distinct}}.
  \end{split}
\end{equation}
We observe that the ring of real valued functions on perfect matchings is isomorphic to
\(\R[\{x_{uv}\}_{\{u,v\} \in \binom{[n]}{2}}] / \ispan{\cP_{n}}\)
with \(x_{uv}\) representing the indicator function
of the edge \(uv\) being contained in a perfect matching.
Intuitively, under this representation the vanishing of the first set of polynomials ensures that no vertex is matched more
than once, the vanishing of the second set ensures that each vertex is matched, and the
vanishing of the third set ensures that each edge coordinate is 0-1 valued. 

Now we formulate low-degree derivations.
Let $\mathcal{P}$ denote a set of polynomials in
$\mathbb{R}[x]$. For polynomials $F$ and $G$, we write $F
\same{(\mathcal{P},d)} G$, or \emph{$F$ is congruent to $G$ from
$\mathcal{P}$ in degree $d$}, if and only if there exist polynomials
$\{q(p): p \in \mathcal{P}\}$ such that
\[F + \sum_{p \in \mathcal{P}} q(p)\cdot p = G\]
and $\max_{p} \deg(q(p) \cdot p) \leq d$.
We often drop the dependence
on $\mathcal{P}$ when it is clear from context.
We shall write \(F \feq G\) for two polynomials \(F\) and \(G\)
defining the same function on perfect matchings,
i.e., \(F - G \in \ispan{\cP_{n}}\).
(Note that as \(\cP_{n}\) contains \(x_{uv}^{2} - x_{uv}\)
for all variables \(x_{uv}\),
the ideal generated by \(\cP_{n}\) is automatically radical.)

\subsection{Deriving that symmetrized polynomials are constant}
\label{sec:sym-poly-constant}

Averaging any polynomial on matchings over the symmetric group gives a constant. In this section we show that this fact has a low-degree derivation.

For a partial matching \(M\),
let \(x_{M} \coloneqq \prod_{e \in M} x_{e}\)
denote the product of edge variables for the edges in \(M\).
The first step is to reduce every polynomial to a linear combination
of the \(x_{M}\).
\begin{lemma}
  \label{lem:monomials}
  For every polynomial \(F\) there is a polynomial
  \(F'\) with \(\deg F' \leq \deg F\)
  and \(F \same{(\mathcal{P}_n,\deg F)} F'\),
  where all monomials of \(F'\) have the form \(x_{M}\)
  for some partial matching \(M\).
\end{lemma}
\begin{proof}
It suffices to prove the lemma when \(F\) is a monomial. Let
\(F = \prod_{e \in A} x_{e}^{k_{e}}\)
for a set \(A\) of edges with multiplicities \(k_{e} \geq 1\).
From \(x_{e}^{2} \same{2} x_{e}\) it follows that
\(x_{e}^{k} \same{k} x_{e}\) for all \(k \geq 1\),
hence \(F \same{\deg F} \prod_{e \in A} x_{e}\).
If \(A\) is a partial matching we are done, otherwise there are distinct \(e, f \in A\) with a common vertex,
hence \(x_{e} x_{f} \same{2} 0\) and \(F \same{\deg F}
0\).
\end{proof}

\begin{lemma}
  \label{lem:matching+a}
  For any partial matching \(M\) on \(2d\) vertices
  and a vertex \(a\) not covered by \(M\),
  we have
  \begin{equation}
    \label{eq:matching+a}
    x_{M}
    \same{(\mathcal{P}_n,d+1)}
    \sum_{\substack{M_{1} = M \cup \{a,u\} \\
        u \in K_{n} \setminus (M \cup \{a\})}}
    x_{M_{1}}
    .
  \end{equation}
\end{lemma}
\begin{proof}
We use the generators \(\sum_{u} x_{au} - 1\)
to add variables corresponding to edges at \(a\),
and then use \(x_{au} x_{uv}\) to remove monomials
not corresponding to a partial matching:
\begin{equation*}
  x_{M}
  \same{(\mathcal{P}_n, d+1)}
  x_{M} \sum_{u \in K_n} x_{au}
  \same{(\mathcal{P}_n,d+1)}
  \sum_{\substack{M_{1} = M \cup \{a,u\} \\
      u \in K_{n} \setminus (M \cup \{a\})}}
  x_{M_{1}}
  .
\end{equation*}
\end{proof}

This leads to a similar congruence using all containing
matchings of a larger size:
\begin{lemma}
  \label{lem:partial-matching}
  For any partial matching \(M\) of \(2d\) vertices
  and \(d \leq k \leq n/2\),
  we have
  \begin{equation}
    \label{eq:partial-matching}
    x_{M} \same{(\mathcal{P}_n,k)}
    \frac{1}{\binom{n/2 - d}{k - d}}
    \sum_{\substack{M' \supset M \\ \size{M'} = k}} x_{M'}
  \end{equation}
\end{lemma}
\begin{proof}
We use induction on \(k - d\).
The start of the induction is with \(k = d\),
when the sides of \prettyref{eq:partial-matching}
are actually equal. If \(k > d\), let \(a\) be a fixed vertex not covered by \(M\).
Applying \prettyref{lem:matching+a} to \(M\) and \(a\)
followed by the inductive hypothesis gives:
\begin{equation*}
    x_{M}
    \same{(\mathcal{P}_{n}, d+1)}
    \sum_{\substack{M_{1} = M \cup \{a,u\} \\
        u \in K_{n} \setminus (M \cup \{a\})}}
    x_{M_{1}}
    \same{(\mathcal{P}_n,k)}
    \frac{1}{\binom{n/2 - d - 1}{k - d - 1}}
    \sum_{\substack{
        M' \supset M_{1} \\ \size{M'} = k \\
        M_{1} = M \cup \{a,u\} \\
        u \in K_{n} \setminus (M \cup \{a\})}}
    x_{M'}
    .
\end{equation*}
Averaging over all vertices \(a\) not covered by \(M\),
we obtain:
\begin{equation*}
  x_{M}
  \same{(\mathcal{P}_n, k)}
  \frac{1}{n - 2 d}
  \frac{1}{\binom{n/2 - d - 1}{k - d - 1}}
  \sum_{\substack{
      M' \supset M_{1} \\ \size{M'} = k \\
      M_{1} = M \cup \{a,u\} \\
      a, u \in K_{n} \setminus M}}
  x_{M'}
  =
  \frac{1}{n - 2 d}
  \frac{1}{\binom{n/2 - d - 1}{k - d - 1}}
  2 (k - d)
  \sum_{\substack{
      M' \supset M \\ \size{M'} = k}}
  x_{M'}
  =
  \frac{1}{\binom{n/2 - d}{k - d}}
  \sum_{\substack{M' \supset M \\ \size{M'} = k}}
  x_{M'}
  .
\end{equation*}
where in the second step the factor \(2(k - d)\) accounts for the number of ways to choose \(a\) and \(u\).
\end{proof}

We are now ready to state and prove the claim about symmetrized polynomials:
\begin{lemma}
  \label{lem:constant}
  For any polynomial \(F\),
  there is a constant \(c_{F}\) with
  \(\sum_{\sigma \in S_{n}} \sigma F \same{(\mathcal{P}_n,\deg F)} c_{F}\).

\begin{proof}
Given \prettyref{lem:monomials},
it suffices to prove the claim for
\(F = x_{M}\) for some partial matching \(M\).
Note that if \(\size{M} = k\) then (using the orbit-stabilizer theorem) the size of the stabilizer of \(M\) is \(2^k k! (n-2k)!\). Now apply \prettyref{lem:partial-matching} with \(d = 0\):
\begin{equation*}
  \sum_{\sigma \in S_{n}} \sigma x_{M}
  = 2^{k} k! (n-2k)! \sum_{M' \colon \size{M'} = k} x_{M'}
  \same{k}
  2^{k} k! (n-2k)! \binom{n/2}{k}.
\end{equation*}
\end{proof}
\end{lemma}

\subsection{Low-degree certificates for matching ideal membership}
\label{sec:derivation}

In this section we present a crucial part of our argument, namely that every degree-$d$ polynomial that is identically zero
over perfect matchings has a derivation of this fact whose degree is $O(d)$.

The following lemma will allow us to apply induction:
\begin{lemma}
  \label{lem:degree-increase}
  If \(L\) is a polynomial with
  \(L \same{(\mathcal{P}_{n-2}, d)} 0\) for some \(d\),
  and \(a, b\) are the two additional vertices in \(K_n\),
  then \(L x_{ab} \same{(\mathcal{P}_{n}, d + 1)} 0\).
\end{lemma}
\begin{proof}
It is enough to prove the claim for \(L \in \mathcal{P}_{n-2}\).
For \(L = x_{e}^{2} - x_{e}\) and \(L = x_{uv} x_{uw}\)
the claim is trivial since \(L \in \mathcal{P}_{n}\) also.
The remaining case is \(L = \sum_{u \in K_{n-2}} x_{uv} - 1\)
for some \(v \in K_{n-2}\).
Then
\[L x_{ab} = 
\left( \sum_{u \in K_{n-2}} x_{uv} - 1 \right) x_{ab} =
\left( \sum_{u \in K_{n}} x_{uv} - 1 \right) x_{ab}
- x_{av} x_{ab} - x_{bv} x_{ab} \same{d+1} 0.
\]
The degree of the derivation is at most \(d+1\) since we can simply multiply the degree-\(d\) derivation for \(L \same{} 0\) by \(x_{ab}\).
\end{proof}

We now show that any \(F \in \ispan{\cP_{n}}\) can be generated
by low-degree coefficients from \(\cP_{n}\):
\begin{theorem} \label{thm:derivation}
  \label{thm:main2}
For every polynomial $F \in \R[\{x_{uv}\}_{\{u,v\} \in
  \binom{n}{2}}]$, if $F \in \ispan{\cP_{n}}$
  then $F \same{(\cP_n, 2 \deg F - 1)} 0$.

\begin{proof}
We use induction on the degree \(d\) of \(F\).
If \(d = 0\) then \(F = 0\) and the statement holds trivially.
(Note that \(\same{-1}\) is just equality.)
The case \(d = 1\) rephrased means that
the affine space spanned by the characteristic vectors of
all perfect matchings
is defined by the \(\sum_{v} x_{uv} - 1\) for all vertices \(u\).
This follows from Edmonds's description
of the perfect matching polytope by linear inequalities
in \cite{Edmonds65}.

For the case \(d \geq 2\) we first prove the following claim:
\begin{claim*}
If \(F \in \ispan{\cP_n}\) is a degree-\(d\) polynomial and \(\sigma \in S_n\) is a permutation of vertices, then
\[F \same{(\mathcal{P}_n,2d-1)} \sigma F.\]
\end{claim*}
We use induction on the degree. If \(d = 0\) or \(d = 1\) the claim follows from the corresponding cases \(d = 0\) and \(d = 1\) of the theorem.
For \(d \ge 2\) it is enough to prove the claim
when \(\sigma\) is a transposition of two vertices \(a\) and \(u\).
Note that in \(F - \sigma  F\)
all monomials which are independent of both \(a\) and \(u\) cancel:
\begin{equation}
\label{eq:transposition-diff-sum}
F - \sigma F = \sum_{e \colon a \in e \text{ or } u \in e} L_{e} x_{e}
\end{equation}
where each \(L_{e}\) has degree at most \(d - 1\).
We now show that every summand is congruent to a sum of monomials
containing edges incident to both \(a\) and \(u\). For example, for \(e=\{a,b\}\) in \prettyref{eq:transposition-diff-sum} we apply the generator \(\sum_v x_{uv} - 1\) to find:
\begin{equation*}
 L_{ab} x_{ab}
 \same{d+1} L_{ab} x_{ab} \sum_{v} x_{uv}
 \same{d+1} \sum_v L_{ab}x_{ab}x_{uv}.
\end{equation*}
Therefore
\begin{equation*}
  F - \sigma F \same{d+1} \sum_{bv} L'_{bv} x_{ab} x_{uv}
\end{equation*}
for some polynomials \(L'_{bv}\) of degree at most \(d-1\).
We may assume that \(L'_{bv}\) does not contain variables \(x_{e}\)
with \(e\) incident to \(a, b, u, v\),
as these can be removed
using generators like \(x_{ab} x_{ac}\) or \(x_{ab}^{2} - x_{ab}\).
Moreover, it can be checked that \(L'_{bv}\) is zero on all perfect matchings containing
\(\{a, b\}\) and \(\{u,v\}\).
By induction, \(L'_{bv} \same{(\cP_{n-4}, 2d-3)} 0\)
(identifying \(K_{n-4}\)
with the graph \(K_{n} \setminus \{a, b, u, v\}\)),
from which \(L'_{bv} \same{(\cP_{n}, 2d - 1)} 0\) follows
by two applications of \prettyref{lem:degree-increase}. (The special case
\(a = v, b = u\) is also handled by induction and one application of \prettyref{lem:degree-increase}.) This concludes the proof of the claim.

We now apply the claim followed by \prettyref{lem:constant}:
\begin{equation*}
  F \same{2d-1} \frac{1}{n!} \sum_{\sigma \in S_{n}} \sigma F
  \same{d} \frac{c_{F}}{n!}
\end{equation*}
for a constant \(c_{F}\).
As $F \in \ispan{\cP_n}$, it must be that \(c_{F} = 0\), and therefore
\(F \same{2d - 1} 0\).
\end{proof}
\end{theorem}

\subsection{The main theorem}

We now have all the ingredients to prove our main theorem.
Note that the alternating group \(A_{n}\) acts naturally
on \(\perfectmatching{n}\) via permutation of vertices.
Recall that we set \(\tilde{S}(f) \coloneqq \max f\)
and \(\tilde{C}(f) \coloneqq \max f + \varepsilon / 2\), where the functions \(f\) are indexed by edge set and \(\varepsilon\) is a parameter. It follows that the guarantees \(\tilde{C}, \tilde{S}\) are
\(A_{n}\)-symmetric in the sense defined in \prettyref{sec:sym-sdp}.
Our main theorem is an exponential lower bound on the size
of any \(A_{n}\)-coordinate-symmetric SDP extension of
\(\perfectmatching{n}\).
\begin{theorem}[Main]
  \label{thm:matching}
  There exists an absolute constant $\alpha > 0$
  such that for all even $n$ and every \(0 \leq \varepsilon < 1\),
  every \(A_{n}\)-coordinate-symmetric SDP extended formulation
  approximating the perfect matching problem $\perfectmatching{n}$
  within a factor of \(1 - \varepsilon / (n - 1)\)
  has size at least \(2^{\alpha n}\).
  
\begin{proof}
Fix an even integer $n \geq 10$ and
let \(k = \lceil \beta n \rceil\)
for some small enough constant \(0 < \beta < 1/2\) chosen later.
Suppose for a contradiction
that $\perfectmatching{n}$ admits a symmetric SDP extended
formulation of size \(d < \sqrt{\binom{n}{k}} - 1\).

Let \(m\) equal \(n/2\) or \(n/2 - 1\), whichever is odd. 
Let \(S = [m]\) and let \(T = \{m + 1,\ldots, 2m\}\). If \(m = n/2\) then let \(U = \{2m+1, 2m+2\}\), otherwise let \(U = \varnothing\).
Note that \(S \cup T \cup U = [n]\) and \(\size{S} = \size{T} = m = \Theta(n)\).
Consider the objective function for the set of edges \(E[S]\), namely \(f_{E[S]}(M) \coloneqq \card{M \cap E[S]}\). Since \(\size{S}\) is odd we have \(\max f_{E[S]} = (\size{S} - 1) / 2\), from which we obtain:
\begin{equation}
  \label{eq:matching-failed}
 f(x) \defeq \tilde{C}(f_{E[S]}) - f_{E[S]}(x)
= \frac{\size{S} - 1}{2} + \frac{\varepsilon}{2}
- \sum_{u, v \in S} x_{uv} \feq
 \frac{1}{2}\sum_{u \in S, v \in T \cup U} x_{uv} - \frac{1 - \varepsilon}{2}.
\end{equation}
By \prettyref{lem:factorSymSDP},
as \(\binom{d + 1}{2} < \binom{n}{k}\),
there is a constant \(\mu_{f} \geq 0\) and
an \(A_{n}\)-symmetric set \(\cH\) of functions
of size at most \(\binom{n}{k}\)
on the set of perfect matchings
with
\[ f \feq \sum_{g} g^2 + \mu_{f} \qquad \text{with each }
g \in \vspan{\cH}.\]
By \prettyref{prop:matching-junta},
every \(h \in \cH\) depends only on the edges within a vertex set of size less than \(k\), and hence can be represented by a polynomial of degree less than \(k/2\) over perfect matchings.
As the \(g\) are linear combinations of the \(h \in \cH\),
they can also be represented by polynomials of degree less than
\(k/2\), which we assume for the rest of the proof.

Applying \prettyref{thm:derivation} with \eqref{eq:matching-failed},
we conclude
\begin{equation*}
  \frac{1}{2}\sum_{u \in S, v \in T \cup U} x_{uv} - \frac{1 - \varepsilon}{2}
  \same{(\cP_{n}, 2 k - 1)} \sum_{g} g^{2} + \mu_f.
\end{equation*}
We now apply the following substitution: set \(x_{2m+1,2m+2} \seteq 1\) if \(U\) is not empty, set \(x_{u+m,v+m} \seteq x_{uv}\)
for each \(uv \in E[S]\), and set \(x_{uv} \seteq 0\) otherwise. Intuitively, the substitution ensures that \(U\) is matched, ensures the matching on \(T\) is identical to the matching on \(S\), and ensures every edge is entirely within \(S\), \(T\), or \(U\). The main point is that the substitution maps every polynomial in \(\cP_{n}\) either to \(0\) or into \(\cP_{m}\).

Applying this substitution we obtain a new polynomial identity on the
variables $\{x_{uv}\}_{\{u,v\} \in \binom{S}{2}}$:
\begin{equation}
\label{eq:sos-refutation-pm}
-\frac{1 - \varepsilon}{2} \same{(\cP_{m}, 2 k - 1)} \sum_{g} g^{2} + \mu_f.
\end{equation}

This equation is a sum of squares proof that an odd clique of size $m$ cannot have a perfect matching. To complete our argument we appeal to a theorem from \cite{grigoriev2001linear} which shows that any such proof must have high degree. Since the degree of the proof in \prettyref{eq:sos-refutation-pm} is \(2k - 1\), our conclusion will be that \(k\) must be large.

The theorem from \cite{grigoriev2001linear} uses different terminology from what we have developed here. It is phrased in terms of Positivstellensatz Calculus (\(PC_>\)) proof systems and the \(\text{MOD}_2\) principle. 
We first present the theorem as originally stated and then relate it to our setting.

\begin{theorem}[{\cite[Corollary~2]{grigoriev2001linear}}]
\label{thm:grigoriev}
The degree of any \(PC_>\) refutation of \(\text{MOD}_2^k\) is greater than \(\Omega(k)\).
\end{theorem}

The \(\text{MOD}_p^k\) principle states that it is not possible to partition a set of size \(k\) into groups of size \(p\) if \(k\) is congruent to \(1\) modulo \(p\). In our case, with \(p = 2\) and \(k\) odd, this is equivalent to the statement that no perfect matching exists in an odd clique.

Likewise, via \cite[Definition~2]{grigoriev2001linear} one checks that
\prettyref{eq:sos-refutation-pm} constitutes a \(PC_>\) proof;
we refer the reader to \cite{buss1999linear} for further discussion.

Applying \prettyref{thm:grigoriev}
to \prettyref{eq:sos-refutation-pm},
we find that \(2 k - 1 = \Omega(m) = \Omega(n)\),
a contradiction when \(\beta\) is chosen small enough.
Since \(\tilde{S}(f) = \max f \leq (n - 1)/2\) when \(f\) is associated with an odd set, we have \((1 - \varepsilon / (n - 1)) \tilde{C}(f) \geq \tilde{S}(f)\), which establishes an inapproximability ratio of \(1 - \varepsilon / (n - 1)\).
\end{proof}
\end{theorem}

\section{The Metric Traveling Salesperson Problem (TSP) revisited}

In this section, we prove that a particular Lasserre SDP is optimal
among all symmetric SDP relaxations for the asymmetric metric
traveling salesperson problem on \(K_{n}\).
The \emph{feasible solutions} of the problem
are all permutations $\sigma\in S_{n}$.
A permutation $\sigma$ corresponds to the tour in $K_n$
in which vertex $i$ is the $\sigma(i)$-th
vertex visited.
An \emph{instance} $\mathcal{I}$ of TSP is a set of non-negative
distances
$d_\mathcal{I}(i,j)$ for each edge $(i,j)\in K_n$, obeying
the triangle inequality. The value of a tour $\sigma$ is just the sum
of the distances of edges traversed
$\val_{\mathcal{I}}(\sigma) = \sum_i d_\cI(\sigma^{-1}(i),\sigma^{-1}(i + 1))$.
The \emph{objective functions} are all the \(\val_{\mathcal{I}}\). Note that TSP is a \emph{minimization} problem rather than a maximization problem, but the framework presented in \prettyref{sec:sym-sdp} generalizes naturally to minimization problems by just flipping the inequalities. For approximation guarantees we will use \(\tilde S(f) = \min f\)
and \(\tilde C(f) = \min f / \rho\) for some factor \(\rho \geq 1\).
Therefore instead of referring to a ``\((\tilde C,\tilde S)\)-approximate formulation''
we will refer to a ``formulation within a factor \(\rho\).''

The natural action of \(A_{n}\) on TSP is by permutation of vertices,
which means here that \(A_{n}\) acts on \(S_{n}\) by composition
from the left:
\((\sigma_{1} \cdot \sigma_{2})(i) = \sigma_{1} (\sigma_{2}(i))\).
Obviously, the problem TSP is \(A_{n}\)-symmetric.

The ring of real-valued functions on the set \(S_{n}\)
of feasible solutions is isomorphic to
\(\R[\{x_{ij}\}_{\{i,j\} \in [n]}] / \ispan{\cQ_{n}}\),
with \(x_{ij}\) being the indicator of \(\sigma(i) = j\),
and \(\cQ_{n}\) is the set of
\emph{TSP constraints}:
\begin{align*}
\mathcal{Q}_n = &\setcomp{\sum_{i \in [n]} x_{ij} - 1}{j \in [n]} \cup \setcomp{\sum_{j \in [n]} x_{ij} - 1}{i \in [n]} \\ 
&\cup \setcomp{x_{ij}x_{ik}}{i,j,k \in [n]} \cup \setcomp{x_{ij}x_{kj}}{i,j,k \in [n]}\\
&\cup \setcomp{x_{ij}^2 - x_{ij}}{i,j \in [n]}.
\end{align*}

We emphasize that our description of the TSP constraints is different from the TSP polytope treated in \cite{Yannakakis91, Yannakakis88} and \cite{extform4, DBLP:journals/jacm/FioriniMPTW15}:
 the variables \(x_{ij}\) do not directly encode the edges of a Hamiltonian cycle but instead specify a permutation of \(n\) vertices, encoded as a perfect bipartite matching on \(K_{n,n}\).

Following the framework presented in \cite{Lee2014power}, we define the Lasserre hierarchy for TSP as follows.
The (dual of) the $k$-th level Lasserre SDP relaxation for a TSP instance $\cI$ is given by
\begin{align*}
\text{Maximize }  & C\\
\text{subject to } & \val_{\cI} - C \same{(\mathcal{Q}_n,k)}
 \sum_{p} p^{2}
 &\text{for some polynomials \(p\)}.
\end{align*}

We now state our main theorem regarding optimal SDP relaxations for
TSP.
\begin{theorem}
\label{thm:tsp-main}
  Suppose that there is some coordinate
  $A_{2n}$-symmetric SDP relaxation of size
  $r < \sqrt{\binom{n}{k}} - 1$
  approximating TSP within some factor \(\rho \geq  1\)
  for instances on $2n$ vertices.
  Then the $(2k-1)$-level Lasserre relaxation approximates
  TSP within the factor of \(\rho\) on instances on $n$ vertices.
\end{theorem}

To prove \prettyref{thm:tsp-main} there is an equivalent of \prettyref{prop:matching-junta} we need for TSP tours, so that a small set of invariant functions depends only on the positions of a small number of indices. We start with the following proposition.

\begin{proposition}
\label{prop:tsp-junta}
  Let \(\cH\) be an \(A_{n}\)-symmetric set of functions of size
  $\binom{n}{k}$
  on the set of TSP tours $\sigma\in S_{n}$.
  Then for every \(h \in \cH\) there is a set
  \(W \subseteq [n]\) of size less than \(k\),
  such that $h(\sigma)$ depends only on the positions
  of the vertices in $W$ in the tour $\sigma$, and
  the sign of $\sigma$ as a permutation.
\end{proposition}

\ifnum\longform=1
	\begin{proof}
For every \(h\in\cH\) we can apply 
\prettyref{lem:DixonMortimer} to
the stabilizer of \(h\) to 
obtain a subset \(W \subseteq [n]\) of size
at most \(k\)
such that \(h\) is stabilized by \(A([n] \setminus
W)\). Thus for every tour $\sigma$, $h$ is constant on the \(A([n] \setminus W)\)-orbit of $\sigma$. This orbit is clearly determined by the positions of the vertices in $W$ and, since $A([n] \setminus W)$ preserves signs, the sign of the permutation $\sigma$. 
\end{proof}

\else
	The proof appears in \prettyref{sec:tsp-junta}.
\fi

Next we give a reduction which allows us to eliminate the dependence of the functions $h\in\cH$ on the sign of the permutation $\sigma$. In particular we encode every TSP tour $\sigma$ on an $n$-vertex graph as some new tour $\Phi(\sigma)$ in a $2n$-vertex graph, such that $\Phi(\sigma)$ is always an even permutation in $S_{2n}$.

\begin{lemma}
\label{lem:tsp-evenPerm}
Let $\cI$ be an instance of TSP on $K_n$. Then there exists an instance $\cI'$ of TSP on $K_{2n}$ and an injective map $\Phi:S_{n}\to S_{2n}$ such that
\begin{enumerate}
\item $\val_{\cI}(\sigma) = \val_{\cI'}(\Phi(\sigma))$ for all $\sigma\in S_n$.
\item For every tour $\tau\in S_{2n}$ there exists $\sigma\in S_n$ such that $\val_{\cI'}(\Phi(\sigma)) \leq \val_{\cI'}(\tau)$
\item For all $\sigma\in S_n$ the permutation $\Phi(\sigma)$ is even.
\end{enumerate}
\end{lemma}
\begin{proof}
Given a TSP instance $\mathcal{I}$ on $K_n$ we construct a new instance $\mathcal{I'}$ on $K_{2n}$ as follows:
\begin{itemize}
\item For every vertex $i \in \mathcal{I}$ add a pair of vertices $i$ and $i'$ to $\mathcal{I'}$.
\item For every distance $d(i,j)$ in $\mathcal{I}$ add 4 edges all with the same distance $d(i,j) = d(i',j) = d(i,j') = d(i',j')$ to $\mathcal{I'}$.
\item For every pair of vertices $i,i' \in \mathcal{I'}$ add an edge of distance zero, i.e.\ set $d(i,i') = 0$.
\end{itemize}
We will call a tour $\tau\in S_{2n}$ \textit{canonical} if it visits $i'$ immediately after $i$, i.e.\ $\sigma(i') = \sigma(i) + 1$. We will write $T$ for the set of canonical tours in $S_{2n}$. It is easy to check using the triangle inequality that for every tour $\tau$ there is a canonical tour with no larger value. For every tour $\sigma$ in $\mathcal{I}$ define $\Phi(\sigma)$ to be the corresponding canonical tour in $\mathcal{I'}$. That is $\Phi(\sigma)(i) = 2\sigma(i) - 1$ and $\Phi(\sigma)(i') = 2\sigma(i)$. Note that $\Phi:S_n\to S_{2n}$ is an injective map whose image is all of $T$. By construction we have:
\[
\val_{\mathcal{I}}(\sigma) \equiv \val_{\mathcal{I'}}(\Phi(\sigma))
\]
which proves property (1). Property (2) follows from the fact that every tour $\tau\in S_{2n}$ has a canonical tour with no larger value, and that $T$ is the image of $\Phi$.

For property (3), note that every canonical tour is an even permutation. To see why, suppose $\sigma\in S_n$ is given by $\sigma = (i_1,j_1)(i_2,j_2),\ldots,(i_m,j_m)$ where $(i,j)$ denotes the permutation that swaps $i$ and $j$. Then $\Phi(\sigma) = (i_1,j_1)(i_1',j_1'),\ldots,(i_m,j_m)(i_m',j_m')$ is comprised of $2m$ swap permutations, and is therefore even.
\end{proof}

The last ingredient we need is a version of \prettyref{thm:derivation}
for the TSP.
\begin{theorem}
\label{thm:tsp-derivation}
If $F$ is a multilinear polynomial whose monomials are partial matchings on $K_{n,n}$ and $F \in \ispan{\cQ_n}$, then $F \same{(\mathcal{Q}_n,2\deg F - 1)} 0$. 
\end{theorem}
Because $\mathcal{Q}_n$ is so similar to $\mathcal{P}_n$, it should come as no surprise that the proof of the above theorem is extremely similar to the proof of \prettyref{thm:main2}. We include the full proof for completeness, but defer it to \prettyref{sec:tour-derivation}. We now have all the tools necessary to prove \prettyref{thm:tsp-main}.
\begin{proof}[Proof of \prettyref{thm:tsp-main}]
First let $\cI$ be an instance of TSP on $K_n$. Use \prettyref{lem:tsp-evenPerm} to construct a TSP instance $\cI'$ on $K_{2n}$ and the corresponding map $\Phi$. Now assume we have an arbitrary $A_{2n}$-symmetric SDP relaxation of size $d < \sqrt{\binom{2n}{k}} - 1$ for TSP on $K_{2n}$. By \prettyref{lem:factorSymSDP} there is a corresponding $A_{2n}$-symmetric family of functions $\mathcal{H'}$ of size $\binom{d+1}{2}$ such that whenever $\min_\tau \val_{\cI'}(\tau) \geq \tilde S(\val_{\cI'})$ we have:
\begin{align*}
\val_{\mathcal{I'}}(\tau) - \tilde C(\val_{\cI'}) \feq \sum_{j} h_j(\tau)^2 + \mu_{\cI'} \qquad \text{where }
                     h_j \in \vspan{\mathcal{H'}} \text{ and } \mu_{\cI'} \geq 0\mper
\end{align*}
Let $h'\in \cH'$. By \prettyref{prop:tsp-junta} $h'(\tau)$ depends only on some subset $W'$ of size at most $k$, and possibly on the sign of $\tau$.

Now we restrict the above relaxation to the image of $\Phi$. By \prettyref{lem:tsp-evenPerm} this does not change the optimum. Using the fact that $\val_{\mathcal{I}}(\sigma) \equiv \val_{\mathcal{I'}}(\Phi(\sigma))$ and setting $\mu_{\cI} = \mu_{\cI'}$ then gives rise to a new relaxation where whenever $\min_\sigma \val_{\cI}(\sigma) \geq \tilde S(\val_{\cI})$ we have:
\[
\val_{\mathcal{I}}(\sigma) - \tilde C(\val_{\cI}) \feq \sum_{j} h_j(\Phi(\sigma))^2 + \mu_{\cI} \qquad \text{where } h_j \in \vspan{\mathcal{H'}}\text{ and }\mu_{\cI} \geq 0
\]
as $\tilde S(\val_{\cI}) = \tilde S(\val_{\cI'})$ and
$\tilde C(\val_{\cI}) = \tilde C(\val_{\cI'})$ by \prettyref{lem:tsp-evenPerm}. Next for each $h'\in\mathcal{H'}$
define $h:S_n \to \R$ by $h(\sigma) = h'(\Phi(\sigma))$. Since
$\Phi(\sigma)$ is even, we then have that each $h$ depends only on the
position of some subset $W \subseteq [n]$ of size at most $k$. Such a
function can be written as a degree-$k$ polynomial $p$ in the
variables $x_{ij}$ so that $p(x^{\sigma}) \feq f(\sigma)$ on the
vertices of $P_{TSP}(n)$. Now by \prettyref{thm:tsp-derivation} we
have that $p \same{(Q_n,2k-1)} h$. Since $\mu_{\cI} \geq 0$ it is clearly the square of a (constant) polynomial, and we conclude that whenever
$\min_\sigma \val_{\cI}(\sigma) \leq \tilde S(\val_{\cI})$ we have:
\begin{align*}
f_{\cI}(x) - \min f_{\cI}/\rho \same{(\mathcal{Q}_n,2k-1)} \sum_{p} p(x)^2
\end{align*}
which is precisely the statement that the $(2k-1)$-level Lasserre relaxation for $P_{TSP}(n)$ is a $\rho$-approximation.
\end{proof}

\ifnum\longform=1
	\ifnum\longform=1
	\subsection{Low-degree certificates for tour ideal membership}
\else
	\section{Low-degree certificates for tour ideal membership}
\fi
\label{sec:tour-derivation}

In this section we prove \prettyref{thm:tsp-derivation}
showing that every degree-$d$ polynomial identically zero
over \emph{TSP tours} is congruent to \(0\) within degree $O(d)$.

Note that any partial tour $\tau$ can be thought of as a partial matching $M$ in $K_{n,n}$, 
namely if $\tau(i) = j$, then $M$ includes the edge $(i,j)$. Because of this, it 
will come as no surprise that the proof proceeds in a very similar
manner to \prettyref{sec:derivation}, and hereafter we shall always refer to partial matchings on $K_{n,n}$
rather than on $K_n$. 

For a partial matching \(M\),
let \(x_{M} \coloneqq \prod_{e \in M} x_{e}\)
denote the product of edge variables for the edges in \(M\).
The first step is to reduce every polynomial to a linear combination
of the \(x_{M}\).
\begin{lemma}
  \label{lem:tour-monomials}
  For every polynomial \(F\) there is a polynomial
  \(F'\) with \(\deg F' \leq \deg F\)
  and \(F \same{(\mathcal{Q}_n,\deg F)} F'\),
  where all monomials of \(F\) have the form \(x_{M}\)
  for some partial matching \(M\).
\end{lemma}
\begin{proof}
It is enough to prove the lemma when \(F\) is a monomial:
\(F = \prod_{e \in A} x_{e}^{k_{e}}\)
for a set \(A \subseteq E[K_{n,n}]\) of edges with multiplicities \(k_{e} \geq 1\).
From \(x_{e}^{2} \same{2} x_{e}\) it follows that
\(x_{e}^{k} \same{k} x_{e}\) for all \(k \geq 1\),
hence \(F \same{\deg F} \prod_{e \in A} x_{e}\),
proving the claim if \(A\) is a partial matching.
If \(A\) is not a partial matching,
then there are distinct \(e, f \in A\) with a common vertex,
hence \(x_{e} x_{f} \same{2} 0\) and \(F \same{\deg F}
0\).
\end{proof}

The rest of the proof proceeds identically to \prettyref{thm:derivation}, but 
we let the symmetric group act on polynomials slightly differently. If $K_{n,n} = U_n \cup V_n$
is the bipartite decomposition of $K_{n,n}$, then we only let the permutation group
act on the labels of vertices of $U_n$, i.e. $\sigma x_{(a,b)} = x_{(\sigma(a),b)}$.
We show that under this action, symmetrized polynomials are congruent to a constant, which can again be seen
in the same sequence of lemmas:
\begin{lemma}
  \label{lem:tour+a}
  For any partial matching \(M\) on \(2d\) vertices
  and a vertex \(a \in U_n\) not covered by \(M\),
  we have
  \begin{equation}
    x_{M}
    \same{(\mathcal{Q}_n,d+1)}
    \sum_{\substack{M_{1} = M \cup \{a,u\} \\
        v \in V_{n} \setminus (M \cap V_n)}}
    x_{M_{1}}
    .
  \end{equation}
\end{lemma}
\begin{proof}
We use the generators \(\sum_{v} x_{av} - 1\)
to add variables corresponding to edges at \(a\),
and then use \(x_{av} x_{bv}\) to remove monomials
not corresponding to a partial matching:
\begin{equation*}
  x_{M}
  \same{(\mathcal{Q}_n, d+1)}
  x_{M} \sum_{v \in V_n} x_{av}
  \same{(\mathcal{Q}_n,d+1)}
  \sum_{\substack{M_{1} = M \cup \{a,v\} \\
     v \in V_{n} \setminus (M \cap V_n)}}
  x_{M_{1}}
  .
\end{equation*}
\end{proof}

This leads to a similar congruence using all containing
matchings of a larger size:
\begin{lemma}
  \label{lem:partial-tour}
  For any partial matching \(M\) of \(2d\) vertices
  and \(d \leq k \leq n\),
  we have
  \begin{equation}
    \label{eq:partial-tour}
    x_{M} \same{(\mathcal{Q}_n,k)}
    \frac{1}{\binom{n - d}{k - d}}
    \sum_{\substack{M' \supset M \\ \size{M'} = k}} x_{M'}
  \end{equation}
\end{lemma}
\begin{proof}
We use induction on \(k - d\).
The start of the induction is when \(k = d\),
when the sides of Equation~\eqref{eq:partial-tour}
are equal.

If \(k > d\), let \(a \in U_n\) be a fixed vertex not covered by \(M\).
Applying \prettyref{lem:tour+a} to \(M\) and \(a\)
followed by the inductive hypothesis gives:
\begin{equation*}
    x_{M}
    \same{(\mathcal{Q}_{n}, d+1)}
    \sum_{\substack{M_{1} = M \cup \{a,u\} \\
        u \in V_{n} \setminus (M \cap V_n)}}
    x_{M_{1}}
    \same{(\mathcal{Q}_n,k)}
    \frac{1}{\binom{n - d - 1}{k - d - 1}}
    \sum_{\substack{
        M' \supset M_{1} \\ \size{M'} = k \\
        M_{1} = M \cup \{a,u\} \\
        u \in V_{n} \setminus (M \cap V_n)}}
    x_{M'}
    .
\end{equation*}
Averaging over all vertices \(a \in U_n\) not covered by \(M\),
we obtain
\begin{equation*}
  x_{M}
  \same{(\mathcal{Q}_n, k)}
  \frac{1}{n - d}
  \frac{1}{\binom{n - d - 1}{k - d - 1}}
  \sum_{\substack{
      M' \supset M_{1} \\ \size{M'} = k \\
      M_{1} = M \cup \{a,u\} \\
      a \in U_n \setminus (M \cap U_n) \\
      u \in V_{n} \setminus (M \cap V_n)}}
  x_{M'}
  =
  \frac{1}{n - d}
  \frac{1}{\binom{n - d - 1}{k - d - 1}}
  (k - d)
  \sum_{\substack{
      M' \supset M \\ \size{M'} = k}}
  x_{M'}
  =
  \frac{1}{\binom{n - d}{k - d}}
  \sum_{\substack{M' \supset M \\ \size{M'} = k}}
  x_{M'}
  .
\end{equation*}
\end{proof}

\begin{corollary}
  \label{cor:tour-constant}
  For any polynomial \(F\),
  there is a constant \(c_{F}\) with
  \(\sum_{\sigma \in S_{n}} \sigma F \same{(\mathcal{Q}_n,\deg F)} c_{F}\).
\end{corollary}
\begin{proof}
In view of \prettyref{lem:tour-monomials},
it is enough to prove the claim for
\(F = x_{M}\) for some partial matching \(M\)
on \(2 k\) vertices,
which is an easy application of \prettyref{lem:partial-tour}
with \(d = 0\):
\begin{equation*}
  \sum_{\sigma \in S_{n}} \sigma x_{M}
  = (n-k)! \sum_{M' \colon \size{M'} = k} x_{M'}
  \same{k}
  (n-k)! \binom{n}{k}
  .
\end{equation*}
\end{proof}

The next lemma will allow us to apply induction:
\begin{lemma}
  \label{lem:tour-degree-increase}
  If \(L\) is a polynomial with
  \(L \same{(\mathcal{Q}_{n-2}, d)} 0\)
  and \(a, b\) are the additional vertices in \(\cQ_n\)
  then \(L x_{ab}x_{ba} \same{(\mathcal{Q}_{n}, d + 2)} 0\).
\end{lemma}
\begin{proof}
It is enough to prove the claim when \(L\)
is from \(\mathcal{Q}_{n-2}\).
For \(L = x_{e}^{2} - x_{e}\), \(L = x_{uv} x_{uw}\), and \(L = x_{uv}x_{wv}\)
the claim is trivial, as then \(L \in \mathcal{Q}_{n}\).
The remaining cases are
\begin{enumerate*}
\item
  \(L = \sum_{u \in U_{n-2}} x_{uv} - 1\) for some \(v \in V_{n-2}\)
\item \(L = \sum_{v \in V_{n-2}} x_{uv} - 1\) for some $u \in U_{n-2}$
\end{enumerate*}.
We only deal with the first case, as the second one is analogous.
Then
\[L x_{ab}x_{ba} = \left( \sum_{u \in U_{n}} x_{uv} - 1 \right) x_{ab}x_{ba}
- x_{a v} x_{ab}x_{ba} - x_{b v} x_{ab}x_{ba} \same{(\cQ_n,d+1)} 0.
\]
\end{proof}

We are now ready to prove \prettyref{thm:tsp-derivation}. 
\begin{proof}[Proof of \prettyref{thm:tsp-derivation}]
We use induction on the degree \(d\) of \(F\).
The case \(d=0\) is obvious, as then clearly \(F = 0\).
(Note that \(\same{-1}\) is just equality.)
The case \(d=1\) rephrased means that
the affine space spanned by the characteristic vectors of
all perfect matchings
is defined by the \(\sum_{v} x_{uv} - 1\) for all vertices \(u\).
This follows again from Edmonds's description
of the perfect matching polytope by linear inequalities
in \cite{Edmonds65} (valid for any graph
in addition to \(K_{2n}\) and \(K_{n,n}\)).

For the case \(d \geq 2\) we first prove the following claim:
\begin{claim*}
If \(F \in \ispan{\cQ_n}\) is a degree-\(d\) polynomial and \(\sigma \in S_n\) is a permutation of vertices, then
\[F \same{(\mathcal{Q}_n,2d-1)} \sigma F.\]
\end{claim*}
We use induction on the degree. If \(d = 0\) or \(d = 1\) the claim follows from the corresponding cases \(d = 0\) and \(d = 1\) of the theorem.
For \(d \ge 2\) it is enough to prove the claim
when \(\sigma\) is a transposition of two vertices \(a\) and \(u\).
Note that in \(F - \sigma  F\)
all monomials which do not contain an $x_e$ with $e$ incident to $a$ or $u$ on the left cancel:
\begin{equation}
\label{eq:transposition-diff-sum-tsp}
F - \sigma F = \sum_{e \colon e = (a,r) \textbf{ or } e = (u,r)} L_{e} x_{e}
\end{equation}
where each \(L_{e}\) has degree at most \(d - 1\).
We now show that every summand is congruent to a sum of monomials
containing edges incident to both \(a\) and \(u\) on the left. For example, for \(e=\{a,b\}\) in \prettyref{eq:transposition-diff-sum-tsp}, we apply the generator \(\sum_v x_{uv} - 1\) to find:
\begin{equation*}
 L_{ab} x_{ab}
 \same{d+1} L_{ab} x_{ab} \sum_{v} x_{uv}
 \same{d+1} \sum_v L_{ab}x_{ab}x_{uv}.
\end{equation*}
Therefore
\begin{equation*}
  F - \sigma F \same{d+1} \sum_{bv} L'_{bv} x_{ab} x_{uv}
\end{equation*}
for some polynomials \(L'_{bv}\) of degree at most \(d-1\).
We may assume that \(L'_{bv}\) does not contain variables \(x_{e}\)
with \(e\) incident to \(a, u\) on the left or \(b, v\) on the right,
as these can be removed
using generators like \(x_{ab} x_{ac}\) or \(x_{ab}^{2} - x_{ab}\).
Moreover, since $F$ is zero on all perfect matchings, it can be checked that \(L'_{bv}\) is zero on all perfect matchings containing
\(\{a, b\}\) and \(\{u,v\}\).
By induction, \(L'_{bv} \same{(\cQ_{n-4}, 2d-3)} 0\)
(identifying \(K_{n-4}\)
with the graph \(K_{n} \setminus \{a, b, u, v\}\)),
from which \(L'_{bv} \same{(\cQ_{n}, 2d - 1)} 0\) follows
by two applications of \prettyref{lem:tour-degree-increase}. (The special case
\(a = v, b = u\) is also handled by induction and one application of \prettyref{lem:tour-degree-increase}.) This concludes the proof of the claim.

We now apply the claim followed by \prettyref{cor:tour-constant}:
\begin{equation*}
  F \same{2d-1} \frac{1}{n!} \sum_{\sigma \in S_{n}} \sigma F
  \same{d} \frac{c_{F}}{n!}
\end{equation*}
for a constant \(c_{F}\).
As $F \in \ispan{\cQ_n}$, it must be that \(c_{F} = 0\), and therefore
\(F \same{2d - 1} 0\).
\end{proof}

\fi

\clearpage

\addreferencesection
\bibliographystyle{plainnat}
\bibliography{bib/mr,bib/dblp,bib/scholar,bib/bibliography,bib/lpsize,bib/journal}

\def\cprime{$'$} \def\cprime{$'$} \def\cprime{$'$} \def\cprime{$'$}
  \def\cprime{$'$} \def\cprime{$'$} \def\cprime{$'$} \def\cprime{$'$}
  \def\cprime{$'$} \def\cprime{$'$} \def\cprime{$'$}
  \def\polhk#1{\setbox0=\hbox{#1}{\ooalign{\hidewidth
  \lower1.5ex\hbox{`}\hidewidth\crcr\unhbox0}}} \def\cprime{$'$}
  \def\cprime{$'$} \def\cprime{$'$} \def\cprime{$'$} \def\cprime{$'$}
  \def\cprime{$'$} \def\cprime{$'$} \def\cprime{$'$} \def\cprime{$'$}
  \def\cprime{$'$} \def\cprime{$'$}
  \def\cfac#1{\ifmmode\setbox7\hbox{$\accent"5E#1$}\else
  \setbox7\hbox{\accent"5E#1}\penalty 10000\relax\fi\raise 1\ht7
  \hbox{\lower1.15ex\hbox to 1\wd7{\hss\accent"13\hss}}\penalty 10000
  \hskip-1\wd7\penalty 10000\box7} \def\cprime{$'$} \def\cprime{$'$}
  \def\cprime{$'$} \def\cprime{$'$} \def\cprime{$'$} \def\cprime{$'$}
  \def\ocirc#1{\ifmmode\setbox0=\hbox{$#1$}\dimen0=\ht0 \advance\dimen0
  by1pt\rlap{\hbox to\wd0{\hss\raise\dimen0
  \hbox{\hskip.2em$\scriptscriptstyle\circ$}\hss}}#1\else {\accent"17 #1}\fi}
\begin{thebibliography}{29}
\providecommand{\natexlab}[1]{#1}
\providecommand{\url}[1]{\texttt{#1}}
\expandafter\ifx\csname urlstyle\endcsname\relax
  \providecommand{\doi}[1]{doi: #1}\else
  \providecommand{\doi}{doi: \begingroup \urlstyle{rm}\Url}\fi

\bibitem[Bazzi et~al.(2015)Bazzi, Fiorini, Pokutta, and
  Svensson]{DBLP:conf/focs/BazziFPS15}
Abbas Bazzi, Samuel Fiorini, Sebastian Pokutta, and Ola Svensson.
\newblock No small linear program approximates vertex cover within a factor \(2
  - \epsilon\).
\newblock In \emph{Proc.\ FOCS}, pages 1123--1142, 2015.
\newblock \doi{10.1109/FOCS.2015.73}.

\bibitem[Braun and Pokutta(2012)]{BP2011}
G{\'a}bor Braun and Sebastian Pokutta.
\newblock An algebraic approach to symmetric extended formulations.
\newblock In \emph{Proc.\ {ISCO}}, pages 141--152. Springer Berlin Heidelberg,
  2012.
\newblock ISBN 978-3-642-32147-4.
\newblock \doi{10.1007/978-3-642-32147-4_14}.

\bibitem[Braun and Pokutta(2015{\natexlab{a}})]{DBLP:conf/soda/BraunP15}
G{\'{a}}bor Braun and Sebastian Pokutta.
\newblock The matching polytope does not admit fully-polynomial size relaxation
  schemes.
\newblock In \emph{Proc.\ SODA}, pages 837--846, 2015{\natexlab{a}}.
\newblock \doi{10.1137/1.9781611973730.57}.

\bibitem[Braun and Pokutta(2015{\natexlab{b}})]{DBLP:journals/tit/BraunP15}
G{\'{a}}bor Braun and Sebastian Pokutta.
\newblock The matching problem has no fully polynomial size linear programming
  relaxation schemes.
\newblock \emph{{IEEE} Transactions on Information Theory}, 61\penalty0
  (10):\penalty0 5754--5764, 2015{\natexlab{b}}.
\newblock \doi{10.1109/TIT.2015.2465864}.

\bibitem[Braun et~al.(2012)Braun, Fiorini, Pokutta, and Steurer]{BraunFPS12}
G{\'a}bor Braun, Samuel Fiorini, Sebastian Pokutta, and David Steurer.
\newblock Approximation limits of linear programs (beyond hierarchies).
\newblock In \emph{Proc.\ FOCS}, pages 480--489, 2012.

\bibitem[Braun et~al.(2015{\natexlab{a}})Braun, Fiorini, Pokutta, and
  Steurer]{DBLP:journals/mor/BraunFPS15}
G{\'{a}}bor Braun, Samuel Fiorini, Sebastian Pokutta, and David Steurer.
\newblock Approximation limits of linear programs (beyond hierarchies).
\newblock \emph{Math. Oper. Res.}, 40\penalty0 (3):\penalty0 756--772,
  2015{\natexlab{a}}.
\newblock \doi{10.1287/moor.2014.0694}.

\bibitem[Braun et~al.(2015{\natexlab{b}})Braun, Pokutta, and
  Zink]{DBLP:conf/stoc/BraunPZ15}
G{\'{a}}bor Braun, Sebastian Pokutta, and Daniel Zink.
\newblock Inapproximability of combinatorial problems via small {LP}s and
  {SDP}s.
\newblock In \emph{Proc.\ STOC}, pages 107--116, 2015{\natexlab{b}}.
\newblock \doi{10.1145/2746539.2746550}.

\bibitem[Braverman and Moitra(2013)]{bravermanmoitra13}
Mark Braverman and Ankur Moitra.
\newblock An information complexity approach to extended formulations.
\newblock In \emph{Proc.\ STOC}, pages 161--170. ACM, 2013.
\newblock ISBN 978-1-4503-2029-0.
\newblock \doi{10.1145/2488608.2488629}.

\bibitem[Bri\"et et~al.(2013)Bri\"et, Dadush, and Pokutta]{BDP2013}
Jop Bri\"et, Daniel Dadush, and Sebastian Pokutta.
\newblock On the existence of {0/1} polytopes with high semidefinite extension
  complexity.
\newblock In \emph{Proc.\ ESA}, pages 217--228. Springer Berlin Heidelberg,
  2013.

\bibitem[Bri\"et et~al.(2015)Bri\"et, Dadush, and
  Pokutta]{DBLP:journals/mp/BrietDP15}
Jop Bri\"et, Daniel Dadush, and Sebastian Pokutta.
\newblock On the existence of 0/1 polytopes with high semidefinite extension
  complexity.
\newblock \emph{Math. Program.}, 153\penalty0 (1):\penalty0 179--199, 2015.
\newblock \doi{10.1007/s10107-014-0785-x}.

\bibitem[Buss et~al.(1999)Buss, Grigoriev, Impagliazzo, and
  Pitassi]{buss1999linear}
Sam Buss, Dima Grigoriev, Russell Impagliazzo, and Toniann Pitassi.
\newblock Linear gaps between degrees for the polynomial calculus modulo
  distinct primes.
\newblock In \emph{Proc.\ STOC}, pages 547--556, 1999.

\bibitem[Chan et~al.(2013)Chan, Lee, Raghavendra, and
  Steurer]{chan2013approximate}
Siu~On Chan, James~R. Lee, Prasad Raghavendra, and David Steurer.
\newblock Approximate constraint satisfaction requires large {LP} relaxations.
\newblock In \emph{Proc.\ FOCS}, pages 350--359, 2013.

\bibitem[Dixon and Mortimer(1996)]{DixonBook}
John~D. Dixon and Brian Mortimer.
\newblock \emph{Permutation groups}.
\newblock Springer Verlag, 1996.

\bibitem[Edmonds(1965)]{Edmonds65}
Jack Edmonds.
\newblock Maximum matching and a polyhedron with {$0,1$}-vertices.
\newblock \emph{J. Res. Nat. Bur. Standards Sect. B}, 69B:\penalty0 125--130,
  1965.
\newblock ISSN 0160-1741.

\bibitem[Fiorini et~al.(2012)Fiorini, Massar, Pokutta, Tiwary, and
  {de}~Wolf]{extform4}
Samuel Fiorini, Serge Massar, Sebastian Pokutta, Hans~Raj Tiwary, and Ronald
  {de}~Wolf.
\newblock Linear vs. semidefinite extended formulations: Exponential separation
  and strong lower bounds.
\newblock In \emph{Proc.\ STOC}, pages 95--106, 2012.

\bibitem[Fiorini et~al.(2015)Fiorini, Massar, Pokutta, Tiwary, and
  de~Wolf]{DBLP:journals/jacm/FioriniMPTW15}
Samuel Fiorini, Serge Massar, Sebastian Pokutta, Hans~Raj Tiwary, and Ronald
  de~Wolf.
\newblock Exponential lower bounds for polytopes in combinatorial optimization.
\newblock \emph{J. Assoc. Comput. Mach.}, 62\penalty0 (2):\penalty0 17, 2015.
\newblock \doi{10.1145/2716307}.

\bibitem[Goemans(2015)]{DBLP:journals/mp/Goemans15}
Michel~X. Goemans.
\newblock Smallest compact formulation for the permutahedron.
\newblock \emph{Math. Program.}, 153\penalty0 (1):\penalty0 5--11, 2015.
\newblock \doi{10.1007/s10107-014-0757-1}.

\bibitem[Goemans and Williamson(1995)]{GoemansWilliamson95}
Michel~X. Goemans and David~P. Williamson.
\newblock Improved approximation algorithms for maximum cut and satisfiability
  problems using semidefinite programming.
\newblock \emph{J. Assoc. Comput. Mach.}, 42:\penalty0 1115--1145, 1995.
\newblock \doi{10.1145/227683.227684}.

\bibitem[Gouveia et~al.(2011)Gouveia, Parrilo, and
  Thomas]{GouveiaParriloThomas2011}
Jo{\~a}o Gouveia, Pablo~A Parrilo, and Rekha~R Thomas.
\newblock Lifts of convex sets and cone factorizations.
\newblock \emph{Math. Oper. Res.}, 38\penalty0 (2):\penalty0 248--264, May
  2011.

\bibitem[Grigoriev(2001)]{grigoriev2001linear}
Dima Grigoriev.
\newblock Linear lower bound on degrees of positivstellensatz calculus proofs
  for the parity.
\newblock \emph{Theoretical Computer Science}, 259\penalty0 (1):\penalty0
  613--622, 2001.

\bibitem[Kaibel et~al.(2010)Kaibel, Pashkovich, and
  Theis]{KaibelPashkovichTheis10}
Volker Kaibel, Kanstantsin Pashkovich, and Dirk~Oliver Theis.
\newblock Symmetry matters for the sizes of extended formulations.
\newblock In \emph{Proc.\ IPCO}, pages 135--148, 2010.
\newblock \doi{10.1007/978-3-642-13036-6_11}.

\bibitem[Lee et~al.(2014)Lee, Raghavendra, Steurer, and Tan]{Lee2014power}
James~R. Lee, Prasad Raghavendra, David Steurer, and Ning Tan.
\newblock On the power of symmetric {LP} and {SDP} relaxations.
\newblock In \emph{Proc.\ CCC}, pages 13--21, 2014.

\bibitem[Lee et~al.(2015)Lee, Raghavendra, and Steurer]{DBLP:conf/stoc/LeeRS15}
James~R. Lee, Prasad Raghavendra, and David Steurer.
\newblock Lower bounds on the size of semidefinite programming relaxations.
\newblock In \emph{Proc.\ STOC}, pages 567--576, 2015.
\newblock \doi{10.1145/2746539.2746599}.

\bibitem[Pashkovich(2014)]{DBLP:journals/mor/Pashkovich14}
Kanstantsin Pashkovich.
\newblock Tight lower bounds on the sizes of symmetric extensions of
  permutahedra and similar results.
\newblock \emph{Math. Oper. Res.}, 39\penalty0 (4):\penalty0 1330--1339, 2014.
\newblock \doi{10.1287/moor.2014.0659}.

\bibitem[Rothvo{\ss}(2014)]{rothvoss2013matching}
Thomas Rothvo{\ss}.
\newblock The matching polytope has exponential extension complexity.
\newblock In \emph{Proc.\ STOC}, pages 263--272, 2014.

\bibitem[Sherali and Adams(1990)]{SheraliAdams1990}
Hanif~D. Sherali and Warren~P. Adams.
\newblock A hierarchy of relaxations between the continuous and convex hull
  representations for zero-one programming problems.
\newblock \emph{SIAM J. Discrete Math.}, 3:\penalty0 411--430, 1990.
\newblock \doi{10.1137/0403036}.

\bibitem[Vandenberghe and Boyd(1996)]{VandenbergheBoyd96}
Lieven Vandenberghe and Stephen Boyd.
\newblock Semidefinite programming.
\newblock \emph{SIAM Rev.}, 38:\penalty0 49--95, 1996.

\bibitem[Yannakakis(1988)]{Yannakakis88}
Mihalis Yannakakis.
\newblock Expressing combinatorial optimization problems by linear programs
  (extended abstract).
\newblock In \emph{Proc.\ STOC}, pages 223--228, 1988.

\bibitem[Yannakakis(1991)]{Yannakakis91}
Mihalis Yannakakis.
\newblock Expressing combinatorial optimization problems by linear programs.
\newblock \emph{J. Comput. System Sci.}, 43\penalty0 (3):\penalty0 441--466,
  1991.
\newblock ISSN 0022-0000.
\newblock \doi{10.1016/0022-0000(91)90024-Y}.

\end{thebibliography}

\clearpage

\ifnum\longform=0
{
	\appendix
	\ifnum\longform=0
	\section{Sum of squares for symmetric SDP formulations}
	\label{sec:sos-sym-sdp-proof}
	In this section we prove \prettyref{lem:factorSymSDP}.
\fi

\begin{proof}
For any psd matrix \(M\) let \(\sqrt{M}\) denote the unique psd matrix
with \(\sqrt{M}^{2} = M\). Note that \(\sqrt{M} \sqrt{M}^{\intercal} = M\) also, since \(\sqrt{M}\) is symmetric.

Let \(\cA\), \(b\), \(\setcomp{X^{s}}{s \in \cS}\), \(\setcomp{w^{f}}{f \in \cF}\) comprise
a \(G\)-coordinate-symmetric SDP formulation of size \(d\).
We define the set \(\cH \coloneqq \Gset{h_{ij}}{i,j \in [d]}\)
via \(h_{ij}(s) \coloneqq \sqrt{X^{s}}_{ij}\).
By the action of \(G\) and the uniqueness of the square root,
we have \(g \cdot h_{ij} = h_{g \cdot i, g \cdot j}\), so \(\cH\) is \(G\)-symmetric.
As \(h_{ij} = h_{ji}\), the set \(\cH\) has at most
\(\binom{d+1}{2}\) elements.

By standard strong duality arguments as in \cite{DBLP:conf/stoc/BraunPZ15},
for every \(f \in \cF\) with \(\max f \leq \tilde{S}(f)\),
there is a \(U^{f} \in \psd{d}\) and \(\mu_{f} \geq 0\)
such that for all \(s \in \cS\),
\[\tilde{C}(f) - f(s) = \Tr[U^{f} X^{s}] + \mu_{f}.\]
Again by standard arguments the trace can be rewritten as a sum of
squares:
\begin{equation*}
\Tr[U^{f} X^{s}]
= \Tr\left[\left(\sqrt{U^{f}}\sqrt{X^{s}}\right)^\intercal
\left(\sqrt{U^{f}}\sqrt{X^{s}}\right)\right]
= \sum_{i,j \in [d]} \left(
    \sum_{k \in [d]} \sqrt{U^{f}}_{ik} \cdot \sqrt{X^{s}}_{kj} \right)^{2}.
\end{equation*}
Therefore
\(\tilde{C}(f) - f = \sum_{i,j \in [d]} \left(
  \sum_{k \in [d]} \sqrt{U^{f}}_{ik} \cdot h_{kj}
\right)^{2}
+ \mu_{f}\),
as claimed.
\end{proof}

\section{Highly symmetric functions are juntas}

\subsection{Proof of \prettyref{prop:tsp-junta}}
\label{sec:tsp-junta}
In this section we argue briefly how \prettyref{lem:DixonMortimer} can be applied in an identical manner to get \prettyref{prop:tsp-junta}:
\begin{proof}
For every \(h\in\cH\) we can apply 
\prettyref{lem:DixonMortimer} to
the stabilizer of \(h\) to 
obtain a subset \(W \subseteq [n]\) of size
at most \(k\)
such that \(h\) is stabilized by \(A([n] \setminus
W)\). Thus for every tour $\sigma$, $h$ is constant on the \(A([n] \setminus W)\)-orbit of $\sigma$. This orbit is clearly determined by the positions of the vertices in $W$ and, since $A([n] \setminus W)$ preserves signs, the sign of the permutation $\sigma$. 
\end{proof}

	\clearpage
}
\fi

\end{document}